\documentclass[reqno,11pt]{amsart}
\usepackage[linktocpage]{hyperref}
\usepackage[utf8]{inputenc}
\usepackage{graphicx}
\usepackage{slashed}
\usepackage{amscd}
\usepackage{amssymb}
\usepackage{esint}
\usepackage{enumitem}
\usepackage{mathtools}
\usepackage{comment}
\usepackage{amsmath}
\usepackage{bbm}
\usepackage{dsfont}

\usepackage[mathscr]{eucal}

\usepackage{pst-coil}
\usepackage{pst-grad}
\usepackage{pst-node, pst-plot}
\usepackage{epsfig}
\usepackage[nomessages]{fp}

\usepackage[space]{grffile} 
\usepackage{etoolbox} 
\makeatletter 
\patchcmd\Gread@eps{\@inputcheck#1 }{\@inputcheck"#1"\relax}{}{}
\makeatother

\numberwithin{equation}{section}
\allowdisplaybreaks[1]

\title[An Integral Representation in the Reissner-Nordstr\"om Geometry]{An Integral Representation for the Dirac Propagator in the Reissner-Nordstr\"om Geometry
in Eddington-Finkelstein Coordinates}

\author[F.\ Finster]{Felix Finster}
\author[C.\ Krpoun]{Christoph Krpoun \\ \\ February 2023 / May 2025} 
\address{Fakultät f\"ur Mathematik \\ Universit\"at Regensburg \\ D-93040 Regensburg \\ Germany}
\email{finster@ur.de, christoph.krpoun@mathematik.ur.de}

\newtheorem{Def}{Definition}[section]

\newtheorem{Thm}[Def]{Theorem}
\newtheorem{Prp}[Def]{Proposition}
\newtheorem{Lemma}[Def]{Lemma}
\newtheorem{Remark}[Def]{Remark}

\newcommand{\Thanks}{\vspace*{.5em} \noindent \thanks}
\newcommand{\beq}{\begin{equation}}
\newcommand{\eeq}{\end{equation}}

\newcommand{\Proof}{\begin{proof}}
\newcommand{\QED}{\end{proof} \noindent}
\newcommand{\QEDrem}{\ \hfill $\Diamond$}

\newcommand{\la}{\langle}
\newcommand{\ra}{\rangle}
\newcommand{\bra}{\mathopen{<}}
\newcommand{\ket}{\mathclose{>}}
\newcommand{\Scpr}[2]{(\, #1 \, | \, #2\,)}

\newcommand{\Jost}[1]{\mathcal{J}_{#1}}
\newcommand{\Sl}{\mathopen{\prec}}
\newcommand{\Sr}{\mathclose{\succ}}

\newcommand{\C}{\mathbb{C}}
\newcommand{\R}{\mathbb{R}}

\newcommand{\Z}{\mathbb{Z}}
\newcommand{\N}{\mathbb{N}}

\renewcommand{\H}{\mathscr{H}}

\newcommand{\mm}{\mycal M}
\newcommand{\nn}{\mycal N}
\newcommand{\Am}{\mathcal{A}}

\newcommand{\bb}{{\mathcal{B}}}

\newcommand{\rr}{\mathcal{R}}

\newcommand{\Dir}{{\mathcal{D}}}
\newcommand{\D}{{\mathscr{D}}}

\renewcommand{\L}{{\mathcal{L}}}

\newcommand{\Cisc}{C^\infty_{\text{sc}}}

\newcommand{\Idmat}{\mathbbm{1}}

\newcommand{\dif}[1]{\text{\rm{d}}#1}

\newcommand{\V}{\noindent}

\DeclareMathOperator{\re}{Re}
\DeclareMathOperator{\im}{Im}

\def\bes#1\ees{\begin{align*} #1 \end{align*}}
\def\be#1\ee{\begin{align} #1 \end{align}}
\def\bs{\begin{split}}
\def\es{\end{split}}

\newcommand{\scrM}{\mycal M}
\newcommand{\scrN}{\mycal N}

\DeclareFontFamily{OT1}{rsfso}{}
\DeclareFontShape{OT1}{rsfso}{m}{n}{ <-7> rsfso5 <7-10> rsfso7 <10-> rsfso10}{}
\DeclareMathAlphabet{\mycal}{OT1}{rsfso}{m}{n}

\DeclareFontFamily{U}{mathx}{\hyphenchar\font45}
\DeclareFontShape{U}{mathx}{m}{n}{
      <5> <6> <7> <8> <9> <10>
      <10.95> <12> <14.4> <17.28> <20.74> <24.88>
      mathx10
      }{}
\DeclareSymbolFont{mathx}{U}{mathx}{m}{n}
\DeclareFontSubstitution{U}{mathx}{m}{n}
\DeclareMathAccent{\widecheck}{0}{mathx}{"71}
\DeclareMathAccent{\wideparen}{0}{mathx}{"75}

\DeclarePairedDelimiter\Bra{\langle}{\rvert}
\DeclarePairedDelimiter\Ket{\lvert}{\rangle}
\DeclarePairedDelimiterX\braket[2]{\langle}{\rangle}{#1 \delimsize\vert #2}

\begin{document}
\maketitle

\begin{abstract}
The Cauchy problem for the massive Dirac equation is studied in the Reissner-Nordstr\"om geometry in horizon-penetrating Eddington-Finkelstein-type coordinates. We derive an integral representation for the Dirac propagator involving the solutions of the ordinary differential equations which arise in the separation of variables. Our integral representation describes the dynamics of Dirac particles outside and across the event horizon, up to the Cauchy horizon.
\end{abstract}
\tableofcontents

\section{Introduction} \label{secintro}
The Dirac equation in curved spacetime
describes the dynamics of quantum mechanical waves in the presence of classical gravitational fields.
The propagation of Dirac waves in black hole geometries is of particular interest with respect to
Hawking radiation, the stability of black holes and the fermionic signature operator.
So far, this problem has been studied mainly in the exterior region of the black hole~\cite{tkerr, decay}.
Here we turn attention to the behavior of Dirac waves across and inside the event horizon.
We consider the Reissner-Nordström geometry, which describes a spherical symmetric
and charged black hole.
Our starting point is the Dirac equation in Eddington-Finkelstein coordinates as derived in~\cite{krpoun-mueller}.
Our main result is to derive a corresponding integral representation of the propagator
\[
    \psi(\tau) = e^{-i\tau H} \psi_0 = \int_{\sigma(H)}e^{-i\omega \tau} \:\dif{E}_{\omega} \psi_0 \:,
\]
where~$\dif{E}_{\omega}$ is the spectral measure of the Dirac Hamiltonian~$H$, and $\psi_0$ is smooth
initial data with compact support. We express the spectral measure explicitly in terms of the
fundamental solutions of the radial ordinary differential equation (ODE) arising in Chandrasekhar's separation of variables
(see Theorem~\ref{mainTheo}).
We remark that similar results have been derived previously in the Kerr geometry~\cite{hamilton}.
The novel feature of the present paper is the charge of the black hole. Moreover, we work out the integral representation
in more detail and simplify the formulas considerably.
Our integral representation will be used in a follow-up paper to compute the spectrum of the fermionic signature operator~\cite{sigrn}.

The paper is organized as follows. In Section~\ref{section1} we give
the necessary preliminaries on the Dirac equation
in globally hyperbolic spacetimes and on its separation in the Reissner-Nordstr\"om geometry
in Eddington-Finkelstein-type coordinates.
In Section~\ref{section3} we introduce Dirichlet-type boundary conditions
inside the Cauchy horizon and show that the resulting Hamiltonian is essentially self-adjoint.
Moreover, we express the spectral measure via Stone's formula in terms of the resolvent.
In Section~\ref{secresolvent} the resolvent is computed in terms of the fundamental solutions.
To this end, we construct Jost solutions and use the conservation law for the radial flux in order to
compute the Green's matrix.
In Section~\ref{secintrep} we use the obtained formulas for the resolvent in order to
express the spectral measure explicitly in terms of the fundamental solutions.
This gives the simple and useful analytic expression for the Dirac propagator as stated
in Theorem~\ref{mainTheo}.

\section{Preliminaries \label{section1}}
\subsection{The Dirac Equation in a Globally Hyperbolic Spacetime}
We begin with preliminaries on the Dirac equation in globally hyperbolic spacetimes, following the
presentation in~\cite{finite}.
Thus, let $(\mm, g)$ be a four-dimensional, smooth, globally hyperbolic Lorentzian spin manifold. For the signature of the metric we use the convention~$(+,-,-,-)$.
As proven in~\cite{bernal+sanchez}, $\scrM$ admits a smooth foliation~$(\scrN_\tau)_{\tau \in \R}$
by Cauchy hypersurfaces.
We denote the corresponding spinor bundle by~$S\mm$. Its fibers~$S_x\mm$ are endowed
with an inner product~$\Sl .|. \Sr_x$ of signature~$(2,2)$.
The smooth sections of the spinor bundle are denoted by $C^{\infty}(\mm, S\mm)$.
Likewise, $C_0^{\infty}(\mm, S\mm)$ are the smooth sections with compact support.
We also refer to these sections as wave functions and usually denote them by~$\psi$ or $\phi$. On the
wave functions, we introduce the Lorentz invariant inner product
\begin{align*}
\braket{\, \cdot \,}{\, \cdot \,}  \::\: &C^{\infty}(\mm, S\mm) \times C_0^{\infty}(\mm, S\mm) \longrightarrow \C , \notag \\
        &\braket{\, \psi \,}{\, \phi \,} := \int_{\mm} \Sl \, \psi \, | \, \phi \, \Sr_x\: \dif{\mu}_{\mm}.
\end{align*}
We consider the Dirac equation for a given mass parameter $m \geq 0$. We write the Dirac equation as
\beq \label{direq}
(\mathcal{D} - m ) \psi = 0\:,
\eeq
where the Dirac operator takes the form
\[ \mathcal{D} = iG^{k}\partial_{k} + \bb \, : \, C^{\infty}(\mm, S\mm) \longrightarrow C^{\infty}(\mm, S\mm)\:, \]
and~$G^k : T_x\mm \longrightarrow L(S_x\mm)$ are the Dirac matrices. They fulfill the anti-commutation relations
\[ 
        \{G^j,G^k\} = 2\,g^{jk}\,\Idmat_{S_x \mm}\:. \]
One can understand this map as a representation of the Clifford multiplication in components of the general Dirac matrices. We
will use the Feynman dagger notation reading $\slashed{\nu} = G^j \nu_{j}$. The connection part of the covariant derivative is summarized
in the term $\bb$. We remark that the Dirac equation can be written alternatively as~$\mathcal{D} = i G^j \nabla_j$,
where~$\nabla$ is the Levi-Civita spin connection on~$S\mm$.
For more details on the Dirac equation in curved spacetimes, we refer to~\cite{finite} or~\cite{lawson+michelsohn}.

Given initial data on a Cauchy surface, the Dirac equation admits unique global solutions.
Choosing compactly supported initial data, due to finite propagation speed, the resulting solution
also has compact support on any other Cauchy surface. Such solutions are referred to
as being {\em{spatially compact}}. The smooth, spatially compact solutions are denoted by~$\Cisc(\mm, S\mm)$.
On such solutions, one has the
scalar product
\beq \label{scalPro}
\Scpr{\psi}{\phi} = \int_{\nn} \Sl \, \psi \, | \, \nu^jG_j \,\phi \Sr_x\: \dif{\mu}_{\nn}(x)\:,
\eeq
where $\nu$ is the future directed-normal on $\nn$
(due to current conservation, the scalar product is
in fact independent of the choice of~$\scrN$; for details see~\cite[Section~2]{finite}).
Forming the completion gives the Hilbert space~$(\H, \Scpr{\cdot}{\cdot})$.

In this paper, we restrict attention to {\em{stationary}} spacetimes, meaning that there is a
Killing field~$K$ which is asymptotically timelike (for the general definition see~\cite{oneill}).
We always choose the foliation~$(\scrN_\tau)_{\tau \in \R}$ such that the Killing field is given by~$K=\partial_\tau$.
In this case, it is useful to write spacetime as a product~$\scrM = \R \times \scrN$.
Moreover, we can write the Dirac equation in the Hamiltonian form
\beq \label{schroedingerEq}
        i\partial_{\tau} \psi = H \psi \qquad \text{with} \qquad H := -(G^{\tau})^{-1} \bigg( i
        \sum_{\alpha=1}^3
        G^\alpha \partial_\alpha + \bb - m \bigg) \:,
\eeq
where the Hamiltonian~$H$ is an operator acting on~$\H$ with dense domain
\[ \D(H)= C^\infty_0(\scrN, S\scrM) \:. \]
\begin{Lemma}
In a stationary spacetime, the Hamiltonian~$H$ with domain~$\D(H)$ is a symmetric operator on~$\H$.
\end{Lemma}
\begin{proof}
Since the scalar product in \eqref{scalPro} is independent of the choice of the Cauchy surface,
we know that for all~$\psi, \phi \in \D(H)$,
\[ 0 = \partial_{\tau}(\psi | \phi) \:. \]
Since the Dirac matrices $G^k$, as well as the normal vector field $\nu$ and the
volume form do not depend on $\tau$, we only need to differentiate the wave functions.
We thus obtain
\[ 0 = (\partial_{\tau} \psi | \phi ) + (\psi | \partial_{\tau} \phi) = -i \Big( (H\psi| \phi) - (\psi| H \phi) \Big) \:. \]
This concludes the proof.
\end{proof}

\subsection{The Dirac Equation in the Reissner-Nordstr\"om Geometry \label{subsec1}}
We work in Eddington Finkelstein coordinates~$(\tau, r, \vartheta , \varphi)$ in the range
$\R \times (0,\infty) \times (0, \pi) \times [0,2\pi)$ as defined in~\cite{chandra, krpoun-mueller}.
In these coordinates, the metric takes the form
\begin{align}
        g = &\dfrac{\Delta}{r^2}\:\dif{\tau}\otimes \dif{\tau} - \bigg[2 - \dfrac{\Delta}{r^2} \bigg]\dif{r}\otimes\dif{r}
        - \bigg[ 1 - \dfrac{\Delta}{r^2}\bigg]\bigg( \dif{\tau}\otimes\dif{r} + \dif{r} \otimes \dif{\tau}\bigg) \notag \\
                &-r^2 \,\dif{\vartheta} \otimes \dif{\vartheta} - r^2\,\sin(\vartheta)^2 \,\dif{\varphi} \otimes \dif{\varphi}
        \label{RSmetric}
\end{align}
with $\Delta \equiv \Delta(r) = r^2 - 2Mr + Q^2$ and $\tau = t + u - r$. Here $u$ is the
Regge-Wheeler coordinate (``tortoise coordinate'') which is defined in terms of $r$ by
\beq
    \dfrac{\dif{u}}{\dif{r}} = \dfrac{r^2}{\Delta} \:. \label{reggeWheeler}
\eeq
The zeros of the function~$\Delta$ denoted by
\[ r_\pm = M \pm \sqrt{M^2-Q^2} \]
describe the event horizon and Cauchy horizon, respectively. The region~$\{r > r_- \}$
outside the Cauchy horizon is a globally hyperbolic spacetime. The interior of the Cauchy
horizon~$\{ r < r_-\}$, however, is not globally hyperbolic because of the spacetime singularity at~$r=0$.
Our spacetime has the topology $\mm \cong \R^2 \times S^2$.

In~\cite{krpoun-mueller} the Dirac equation was computed and separated in a specific gauge
in which the Dirac matrices are in the Weyl representation.
Starting from this representation, it is most convenient to transform the Dirac wave function and the Dirac operator as
\be
\Psi &= D\, \psi \label{dirtrans} \\
\Gamma_{\text{trafo}} \, D  \,( \Dir - m) \, D^{-1} \,\Psi &= (\rr + \Am) \, \Psi = 0
\ee
with the transformation matrices
\[ 
            D := \dfrac{\sqrt{r}}{r_+}
            \begin{bmatrix}
                |\Delta|^{1/2} & 0 & 0 & 0 \\
                0 & r_+ & 0 & 0\\
                0 & 0 & r_+ & 0\\
                0 & 0 & 0 & |\Delta|^{1/2}
            \end{bmatrix} \:,\qquad
            \Gamma_{\text{trafo}} :=  r
            \begin{bmatrix}
                1 & 0 & 0 & 0 \\
                0 & -1 & 0 & 0\\
                0 & 0 & -1 & 0\\
                0 & 0 & 0 & 1
            \end{bmatrix} \:. \]
Here~$\rr$ and $\Am$ are the radial and angular operators given by
\be
    \label{matricesDirac}
    \rr& := \begin{bmatrix}
        irm & 0 & |\Delta|^{1/2}\D_0& 0 \\
        0 & -irm & 0 &|\Delta|^{-1/2}\D_1 \\
        |\Delta|^{-1/2}\D_1 & 0 & -irm & 0 \\
        0 & |\Delta|^{1/2} \D_0 & 0 & irm \\
    \end{bmatrix} \\[0.3em]
    \Am& := \begin{bmatrix}
        0 & 0 & 0 & \L_+ \\
        0 & 0 & -\L_- & 0 \\
        0 & \L_+ & 0 & 0 \\
        -\L_- & 0 & 0 & 0  \\
    \end{bmatrix} \:,
\ee
where the linear operators $\Dir_{0/1}, \, \L_{\pm}  : C^{\infty}(\mm, S\mm) \longrightarrow C^{\infty}(\mm, S\mm)$
have the from
\be
        \D_0 :=& \,- \big(\partial_{\tau} - \partial_r \big)  \label{operators} \\
        \D_1 :=& \, \big( 2r^2 - \Delta \big)\partial_{\tau} + \Delta\partial_r  \\
        \L_{\pm} :=& \, \partial_{\vartheta} + \dfrac{\cot(\vartheta)}{2} \mp i \csc(\vartheta)\partial_{\varphi} \label{Lops} \:.
\ee
In this formulation, the spin inner product takes the form
\be \label{spininner}
\Sl \psi | \phi \Sr_x = - \la \Psi, \begin{pmatrix} 0 & \Idmat_{\C^2} \\ \Idmat_{\C^2} & 0 \end{pmatrix} \Phi \ra_{\C^4} \:.
\ee

Employing the separation ansatz
\be
    \label{waveFunc}
        \Psi = \, e^{-i(k + \frac{1}{2})\varphi}\dfrac{1}{r_+}
        \begin{bmatrix}
            \quad \: X_+(\tau, r)\: Y_l(\vartheta)_+ \\
            r_+ \,  X_-(\tau, r)\: Y_l(\vartheta)_- \\
            r_+ \, X_-(\tau, r)\: Y_l(\vartheta)_+ \\
            \quad \:X_+(\tau ,r) \:Y_l(\vartheta)_- \\
        \end{bmatrix} \quad \text{with} \quad
        \omega \in \R \, \text{ and } \, k, \, l\, \in \Z  \:,
\ee
we obtain the eigenvalue problems
\[ 
        \rr  \Psi = \xi \Psi \quad \text{and} \quad \Am \Psi = - \xi \Psi \]
with a separation constant~$\xi$. In this way, the Dirac equation decouples into a radial and angular ODE
of the form
\be
    \label{radialPDE}
    \begin{bmatrix}
        (2r^2 - \Delta) \, \partial_{\tau} + \Delta \partial_r & |\Delta|^{1/2}(imr - \xi) \\
        -\epsilon(\Delta)|\Delta|^{1/2}(imr + \xi) & - \Delta(\partial_{\tau} - \partial_r)
    \end{bmatrix}
    \begin{pmatrix}
        X_+(\tau, r) \\[0.2em]
        r_+ \, X_-(\tau, r)
    \end{pmatrix} &= 0 \\[0.2em]
    \bigg( \begin{bmatrix}
         0 & \L_- \\
         -\L_+ & 0
    \end{bmatrix} - \xi \, \Idmat_{\C^4} \bigg)
    \begin{pmatrix}
        Y_l(\vartheta)_+ \\[0.2em]
        Y_l(\vartheta)_-
    \end{pmatrix} &= 0 \: \label{angularODE}
\ee
where
\be \label{epsdef}
    \epsilon(x) =
        \begin{cases}
            1 \, &\text{if} \quad x > 0\\
            0 \, &\text{if} \quad x = 0\\
            -1 \, &\text{if} \quad x < 0
        \end{cases}
\ee
is the sign function. In view of~\eqref{Lops}, the angular operator in~\eqref{angularODE} does not involve~$\tau$-derivatives
(note that this is no longer the case in the Kerr geometry).
This angular operator is an essentially selfadjoint operator on~$L^2(S^2, \C^2)$ with
dense domain~$C^\infty(S^2, \C^2)$, having a purely discrete spectrum
(for details see~\cite[Section~3 and Appendix~A]{tkerr}). More specifically,
the angular operator is the spin-weighted spherical operator for~$s=\frac{1}{2}$ as analyzed in~\cite{goldberg}.
We denote the corresponding orthonormal eigenvector basis
by $_{\frac{1}{2}} Y_{kl} = e^{-i(k + \frac{1}{2})\varphi} \,Y_{kl}(\vartheta)$ with $l,k \in \Z$, i.e.\
\be \label{angularON}
        \big\la e^{-i(k+\frac{1}{2})\varphi}\, Y_{kl}(\vartheta),e^{-i(k'+\frac{1}{2})\varphi}\, Y_{k'l'}(\vartheta) \big\ra_{L^2(S^2)}
        = \delta_{k,k'}\:\delta_{l,l'} \:.
\ee
Restricting attention to one angular momentum mode,
it suffices to solve the PDE~\eqref{radialPDE} in~$\tau$ and~$r$ for~$\xi=\lambda_{kl}$,
where~$\lambda_{kl}$ denotes the corresponding angular eigenvalue.
Since the angular eigenfunctions are orthonormal, the general Cauchy problem can be solved
by decomposing the initial data into angular momentum modes, solving the PDE~\eqref{radialPDE}
for each mode and taking the superposition.

Similar to~\eqref{schroedingerEq}, the PDE~\eqref{radialPDE} can again be written in the
Hamiltonian form:
\begin{Lemma}
\label{hamilton}
The radial equation~\eqref{radialPDE} can be written in the Schrödinger-type form
    \beq \label{2dirac}
    i \partial_{\tau} X(\tau, r) = H_\xi\, X(\tau, r)
    \eeq
    with the Hamiltonian~$H_\xi$ given by
\[ H_\xi := \underbrace{\begin{bmatrix}
- \frac{i}{(2r^2 - \Delta)} & 0 \\
            0 & \frac{i}{\Delta} \\
        \end{bmatrix}}_{=: \,C^{-1}(r) \, \in \, M(2,\C)}
        \:
        \begin{bmatrix}
            \Delta \partial_r & |\Delta|^{1/2}(imr - \xi) \\
            - \epsilon(\Delta)|\Delta|^{1/2}(imr + \xi) & \Delta \partial_r
        \end{bmatrix} \:, \]
    where $X(\tau, r)$ is the radial part of the wave function from \eqref{waveFunc}.
\end{Lemma}
\begin{proof}
Follows by direct computation.
\end{proof} \noindent
The domain of the radial Hamiltonian in~\eqref{2dirac} will be discussed after Lemma~\ref{lemma-chernoff}
below.

Since $\partial_\tau$ is a Killing field we can write $X(\tau, r) = e^{-i\omega \tau} X^{(\omega)}(r)$ and get a
first-order radial ODE given by
\be
    \label{radialODE2}
        \begin{bmatrix}
            -(2r^2 - \Delta)i \omega + \Delta \partial_r & |\Delta|^{1/2}(imr - \xi) \\
            -\epsilon(\Delta)\,|\Delta|^{1/2}(imr + \xi) & \Delta(i\omega + \partial_r)
        \end{bmatrix}
        \begin{pmatrix}
            X_+^{(\omega)}(r) \\[0.2em]
            r_+ \, X_-^{(\omega)}(r)
        \end{pmatrix} &= 0 \: .
\ee

More details on the previous steps and the asymptotics of the solutions to the radial ODE \eqref{radialODE2}
are worked out in \cite{krpoun-mueller}. We here state the main results,
which can be obtained from~\cite[Theorem~1.1]{krpoun-mueller} by setting the angular momentum~$a$
equal to zero.
\begin{Lemma}
    \label{krpMul}
In the case $ |\omega| < m$, one solution of the radial ODE~\eqref{radialODE2} has exponential decay and the other one exponential growth
    for $u \rightarrow \infty$.

In the case $|\omega| > m$, on the other hand, the solutions have following asymptotics:
\begin{enumerate}[leftmargin=2em]
\item[{\rm{(i)}}] {\bf (Asymptotics at infinity)} Let $w_1 \in \C$ be the root of $\omega^2 - m^2$ contained in the convex hull of $\R_+ $ and $\R_+ \cdot i$ and $w_2 = -w_1$
        the other root, and let
\[ \Theta := \frac{1}{4}\, \ln \Big( \frac{\omega - m}{\omega + m} \Big) \:, \]
then there is~$f_{\infty} := (f_{\infty}^{(1)}, f_{\infty}^{(2)})^T \in\R^2 \setminus \{ 0\} $ with
\[ 
        X^{(\omega)}(u)
        = \begin{bmatrix}
            \cosh(\Theta) & \sinh(\Theta) \\
            \sinh(\Theta) & \cosh(\Theta) \\
        \end{bmatrix}	 \begin{pmatrix}
            f_{\infty}^{(1)}e^{i\Phi_+(u)} \\
            f_{\infty}^{(2)}e^{-i\Phi_-(u)} \\
        \end{pmatrix}
        + E_{\infty}(u)  \]
        \V for the asymptotic phases
        \be
        \Phi_{\pm}(u):= w_{1}\, u + M \bigg( \pm 2\omega + \dfrac{m^2}{w_{1}}\bigg) \ln(u)
        \label{asyPhase}
        \ee
        and for an error function $E_{\infty}(u)$ with polynomial decay. More precisely, there is $c \in \R_+$ with
        \bes
        ||E_{\infty}|| \leqslant \dfrac{c}{u}.
        \ees
\item[{\rm{(ii)}}] {\bf (Asymptotics at the Cauchy horizon)} For every non-trivial solution $X$,
\[ X^{(\omega)}(u) =
        \begin{pmatrix}
            h_{r_-}^{(1)}e^{2i \omega u} \\
            h_{r_-}^{(2)}
        \end{pmatrix}
        + E_{r_-}(u)
\]
        with~$h_{r_-} := (h_{r_-}^{(1)}, h_{r_-}^{(2)})^T \in \R^2 \setminus \{0\} $, with $E_{r_-}$ such that for $r$ sufficiently close to $r_-$ and suitable constants $a,b \in \R_+$,
        \bes
        ||E_{-}(u)|| \leq a e^{- b u}.
        \ees
\end{enumerate}
\end{Lemma}

\section{Functional Analytic Preparations} \label{section3}
In this section we bring the scalar product~\eqref{scalPro} into a more explicit form.
Moreover, we set up the Cauchy problem in a way where spectral methods in Hilbert space become applicable.
In order to obtain a unitary time evolution, we must consider a spacetime region
which includes the Cauchy horizon, and we must introduce reflecting boundary conditions
on the timelike surface~$r=r_0<r_-$ (our solution of the Cauchy problem outside the Cauchy horizon
will not depend on the choice of~$r_0$, as will be explained after Lemma~\ref{lemmaResolvente}).
Moreover, we must make sure that the surfaces~$\{\tau = \text{const}\}$
are spacelike. Noting that the metric coefficient~$g_{rr}$ in~\eqref{RSmetric} has a zero at
\[ r_{\min} = \sqrt{M^2 + Q^2} - M < r_- \:, \]
we are led to choosing~$r_0$ in the interval
\beq \label{r0range}
r_{\min} < r_0 < r_- \:.
\eeq
We thus consider the spacetime region~$M \subset \mm $ defined by
\[ M := \{\tau, r> r_0, \vartheta,\varphi \} \quad \text{with timelike boundary} \quad \partial M := \{\tau, r= r_0,\vartheta, \varphi\} \:. \]
This spacetime is foliated by the space-like hypersurfaces~$(N_{\tau})_{\tau \in \R}$ given by
\bes
N_{\tau} := \{\tau = \text{const.}\, , r \geq r_0, \vartheta , \varphi\} \:.
\ees
Each hypersurface has the boundary~$\partial N_{\tau} = \partial M \cap N_{\tau} \simeq S^2$.
These $N_{\tau} \subset \nn_{\tau}$ give rise to a space-like foliation of $M$.
The Killing field~$K=\partial_\tau$ is tangential to~$\partial M$ and is timelike on this hypersurface
(because~$g(K, K) = g_{00} = \Delta(r_0)/r^2 >0$).
We denote the corresponding spinor bundle by~$SM$.

We first compute the scalar product in this spacetime region.
\begin{Lemma} \label{lemmasprod} The scalar product~\eqref{scalPro} can be written as
\be
    \label{skalarprodukt}
    \Scpr{\psi}{\phi} = \int_{r_0}^{\infty} \dif{r} \int_{-\pi/2}^{\pi/2} \dif{\vartheta} \int_0^{2\pi}
    \dif{\varphi}\;\Psi^{\dagger} \, \Gamma \,\Phi
    \sin(\vartheta)
\ee
(where the capital letters always denote the transformed Dirac wave function~\eqref{dirtrans}) with
\beq \label{scalMat}
    \Gamma = \dfrac{r_+}{|\Delta|}
        \begin{bmatrix}
            2r^2 - \Delta & 0 & 0 & 0 \\
            0 & |\Delta| & 0 & 0 \\
            0 & 0 & |\Delta| & 0 \\
            0 & 0 & 0 & 2r^2 - \Delta \\
        \end{bmatrix} .
\eeq
\end{Lemma} \noindent
Note that, in view of the lower bound in~\eqref{r0range}, the matrix~$\Gamma$ is strictly positive,
showing that~\eqref{skalarprodukt} is indeed a scalar product.
\begin{proof}
By direct computation, we see that the volume form is
\[ \sqrt{|\text{det}g_{\nn_{\tau}}|} = \sqrt{(2r^2- \Delta)} \:r\,\sin{\vartheta} \:. \]
It remains to compute the combination~$G^j \nu_j$.
The normal~$\nu$ is determined by the four equations
\[ g(\nu,\nu) = 1, \quad g(\nu, \partial_r) = 0, \quad g(\nu, \partial_{\varphi}) = 0 \quad \text{and} \quad g(\nu, \partial_{\vartheta}) = 0 \:. \]
By direct computation, we find
\[ \nu = -i \dfrac{\sqrt{ \Delta - 2r^2}}{r} \partial_{\tau} - i \dfrac{\Delta + r^2}{r \sqrt{\Delta - 2r^2}} \partial_{r} \:. \]
This corresponds to the co-vector
\[ \nu = \dfrac{r}{\sqrt{2r^2 - \Delta}}\: \dif{\tau} \:. \]
Reading off the transformed Gamma-matrix $\overline{G^{\tau}}$ from \eqref{matricesDirac},
we need to transform it back with the relation
\[
    G^\tau = D^{-1}\Gamma_{\rm{trafo}}^{-1}\overline{G^\tau}D \, .
\]
Using the form of the spin inner product~\eqref{spininner}, we obtain
\[ \Sl \psi | \slashed{\nu}\, \phi \Sr \sqrt{|\text{det}g_{\nn_{\tau}}|}
= \Psi^{\dagger} \underbrace{D^{-1} \gamma^0 G^{\tau} D^{-1} r^2}_{=: \Gamma} \Phi \sin(\vartheta)
=  \Psi^{\dagger} \Gamma \Phi\sin(\vartheta) \:, \]
concluding the proof.
\end{proof}

In order to obtain a Cauchy problem with a well-defined, unitary time evolution, we
need to introduce suitable boundary conditions at~$r=r_0$. Following the procedure in~\cite{chernoff},
we introduce the reflecting boundary conditions
\[ 
(\slashed{n} - i)\,\psi\big|_{\partial M} = 0 \:, \]
where $\slashed{n}$ is the inner normal on $\partial M$.
For the Cauchy problem, we set~$N = N_\tau|_{\tau=0}$. We choose initial data in the class
\be
    \label{boundary}
C^\infty_\text{init}(N) := \big\{ \psi_0 \in C^{\infty}_0(N \, , S M) \text{ with }
(\slashed{n} - i)\,(H^p\,\psi_0)\big|_{\partial N} = 0 \text{ for all~$p \in \N$} \big\} \:.
\ee
We denote the Hilbert space generated by these functions
(with the scalar product computed in Lemma~\ref{lemmasprod}) by~$\H_N$.
The following lemma was proved in~\cite{chernoff}.

\begin{Lemma} \label{schroedingerPrep}
For initial data~$\psi_0$ in the class~\eqref{boundary}, the Cauchy problem with boundary conditions
\[  i\partial_{\tau} \psi = H \, \psi\:,\qquad \psi|_{N} = \psi_0\:, \qquad
(\slashed{n} - i)\,\psi\big|_{\partial M} = 0 \]
has a unique, global solution $\psi \in \Cisc(M, S M)$. Evaluating this solution at subsequent times $\tau$ and $\tau'$ gives rise to a unique unitary time evolution operator leaving the domain~$C^\infty_\text{\rm{init}}(N)$ invariant, i.e.\
\bes
        U^{\tau'; \tau} : C^\infty_\text{\rm{init}}(N) \subset \H_N \longrightarrow C^\infty_\text{\rm{init}}(N) \subset \H_N \:.
\ees
\end{Lemma}

Having a dense domain which is invariant under the time evolution makes it possible to
apply Chernoff's method~\cite{chernoff73} to obtain the following result. More details can be found in~\cite{chernoff}.
\begin{Lemma} \label{lemma-chernoff}
    The Dirac Hamiltonian $H$ in the Reissner-Nordstr\"om geometry in Edding\-ton-Finkelstein coordinates with domain of definition
    \bes
        \Dir(H)=C^\infty_\text{\rm{init}}(N)
    \ees
    is essentially self-adjoint on the Hilbert space~$\H_N$.
\end{Lemma} \noindent
Having specified the domain of the Hamiltonian, one could go through the transformations in Section~2.2
to work out the corresponding domain of the radial Hamiltonian in~\eqref{2dirac}.
Since the details will not be needed for our results, we omit these computations.

For ease in notation, we denote the self-adjoint extension of the Hamiltonian again by~$H$.
By the spectral theorem for self-adjoint operators we can express
the solution of the Cauchy problem for any~$\psi_0 \in \H_N$ as
\[ 
        \psi(\tau) = e^{-i\tau H} \psi_0 = \int_{\sigma(H)}e^{-i\omega \tau} \:\dif{E}_{\omega} \psi_0 \:, \]
where~$\omega \in \sigma(H)$ are the spectral values, and $\dif{E}_{\omega}$ is the corresponding
spectral measure of $H$.
As explained after~\eqref{angularON}, from now on we restrict attention to one
angular momentum mode and consider the corresponding two-component Dirac equation~\eqref{2dirac}.
For the later explicit analysis, it is helpful to rewrite
the spectral measure with the help of Stone's formula.
\begin{Lemma}
    \label{lemmaResolvente}
The solution of the Cauchy problem can be written as
\be
&X(\tau,r) \notag \\
&= \dfrac{1}{2\pi i} \:\lim_{a \rightarrow \infty} \lim_{\varepsilon \searrow 0} \int_{-a}^a e^{-i \omega \tau}
\bigg[(H_\xi - \omega - i\varepsilon)^{-1} - (H_\xi - \omega + i\varepsilon)^{-1} \bigg]
X_0(r) \:\dif{\omega} \:,  \label{Xresvolent}
\ee
where $(H_\xi - \omega \mp i\varepsilon)^{-1}$ are the resolvents of the Dirac Hamiltonian $H_\xi$ in the upper and lower half-planes
and $X_0(r) \in C^\infty_{\text{\rm{init}}}(N)$ is the initial data for a fixed angular mode $k,l$.
\end{Lemma}
\begin{proof}
Using the properties of the spectral measure, we obtain
\bes
    X(\tau,r) &= e^{-i H_\xi \tau} \lim_{a \rightarrow \infty}E_{(-a,a)} X_0(r) \nonumber \\
             &= \dfrac{1}{2} e^{-i H_\xi \tau} \lim_{a \rightarrow \infty} \bigg[
                 E_{(-a,a)} +  E_{[-a,a]}\bigg]X_0(r)
\ees
Applying Stone's formula (see~\cite[Theorem~VII.13]{reed+simon}) gives the result.
\end{proof}
We point out that in this lemma we solved the Cauchy problem in the region~$M = \R^+ \times (r_0, \infty) \times S^2$
which extends behind the Cauchy horizon at~$r=r_-$.
By restricting this solution to the region~$\R^+ \times (r_-, \infty) \times S^2$, we obtain
the solution of the Cauchy problem outside the Cauchy horizon, for initial data on the
hypersurface~$(r_-, \infty) \times S^2$.
Here we make use of the fact that, due to causality, a Dirac wave cannot cross the Cauchy horizon from
inside to outside.

\section{Computation of the Resolvent} \label{secresolvent}
\subsection{The Green's Matrix} \label{sectPrep}
Our next task is to calculate the resolvent~$(H_\xi - \omega \mp i\varepsilon)^{-1}$
appearing in Lemma~\ref{lemmaResolvente} in the upper and lower half planes.
It is most convenient to work with the Hamiltonian for the radial equation~\eqref{radialPDE}
after separation of variables, denoted by~$H_\xi$.

We introduce the abbreviation~$\omega \mp i\varepsilon  =: \omega_\varepsilon$.
Moreover, it is useful for the computations to write the operator~$H_\xi - \omega_\varepsilon \Idmat_{\C^2}$ as
\beq \label{Hxiomega}
H_\xi - \omega_\varepsilon \Idmat_{\C^2} = C(r)^{-1}\, \rr(\partial_r; \, r)
\eeq
with the matrix
\[
    C(r)^{-1} = \begin{bmatrix}
            - \frac{i}{(2r^2 - \Delta)} & 0 \\
            0 & \frac{i}{\Delta} \\
        \end{bmatrix}
\]
and the radial differential operator
\[ \rr(\partial_r; \, r) := \begin{bmatrix}
            \Delta \partial_r - i\omega_\varepsilon(2r^2 - \Delta) & |\Delta|^{1/2}(imr - \xi) \\
            -\epsilon(\Delta)|\Delta|^{1/2}(imr + \xi) & \Delta(i\omega_\varepsilon + \partial_r)\\
        \end{bmatrix}\:. \]
This makes it possible to write the radial equation as
    \be
    \label{radialEq}
        \Delta(r)\: \partial_r X(r) =
        \begin{bmatrix}
            i\omega_\varepsilon(2r^2 - \Delta) & -|\Delta|^{1/2}(imr - \xi) \\
            \epsilon(\Delta)|\Delta|^{1/2}(imr + \xi) & -i \Delta \omega_\varepsilon \\
        \end{bmatrix} X(r)
    \ee
Our strategy is to invert this differential operator with the help of the so-called Green's matrix~$G(r ; r')_{\omega_\varepsilon}$, being a distributional solution to the equation
\be
    \label{eqDistributional}
    \rr(\partial_r ; r) \:G(r ; r')_{\omega_\varepsilon} = \delta(r - r') \Idmat_{\C^2} \:.
\ee
The Green's matrix can be expressed in terms of the fundamental solutions of the radial equations.
We postpone the detailed computations to the next section (Section~\ref{prepJost}).
Here we explain how the resolvent can be computed from the Green's matrix.
\begin{Lemma}
        \label{propResolvent}
        Let $X_0 \in C^\infty_\text{\rm{init}}(N)$ be initial data for a fixed angular mode $k,l$.
        Then the resolvent acting on~$X_0$ can be expressed in terms of the Green's matrix~$G(r ; r')_{\omega_\varepsilon}$ in~\eqref{eqDistributional} as
\[ R_{\omega_\varepsilon}\, X_0(r) = (H_\xi - \omega_\varepsilon \Idmat_{\C^2})^{-1} X_0(r) = \int_{r_0}^{\infty} G(r; r')_{\omega_\varepsilon}\, C(r')\,
            X_0(r')\:\dif{r'} \:, \]
where~$C(r)$ is the matrix
\[ C(r) = \begin{bmatrix}
                i(2r^2 - \Delta) & 0 \\
                0 & -i\Delta \\
                \end{bmatrix} . \]
\end{Lemma}
\begin{proof} By direct computation using~\eqref{eqDistributional}, one verifies that the
operator~$H_\xi - \omega_\varepsilon \Idmat_{\C^2}$ in~\eqref{Hxiomega} has the inverse
\[ \Big(\big( H_\xi - \omega_\varepsilon \Idmat_{\C^2} \big)^{-1} X \Big)(r) = \int_{r_0}^\infty G(r;r')_{\omega_\varepsilon}\: C(r')\: X(r')\: \dif{r'} \:. \]
This gives the result.
\end{proof}

\subsection{The Radial Jost Solutions}
\label{prepJost}
We now define the fundamental solutions which we will later use to compute the Green's matrix.
Our method is to construct Jost solutions~$\mathcal{J}_{\pm}$.
In preparation, we rewrite the first-order radial system of ODEs stemming
from the matrix Dirac equation in \eqref{radialODE2} into two second-order scalar equations,
sometimes referred to as the Jost equation. We consider the three regions
\beq \label{regions}
(r_0, r_-)\:,\qquad (r_-, r_+) \qquad \text{and} \qquad (r_+, \infty)
\eeq
separately. Choosing the Regge-Wheeler coordinate~$u \in \R$ in one of these regions,
the radial solutions~$X^{(\omega)}(r(u))$ satisfy the equations
\be
    \label{secondOrder}
        \bigg[\partial_u^2 + E_{\omega, l}(u)\partial_u + K^+_{\omega, l}(u)\bigg] X^{(\omega)}_{+}(u) &= 0 \\[0.2em]
        \bigg[\partial_u^2 + E_{\omega, l}(u)\partial_u + K^-_{\omega, l}(u)\bigg] X^{(\omega)}_{-}(u) &= 0
\ee
where $E_{\omega, l}(u)$ and $K^{\pm}_{\omega, l}(u)$ are smooth functions.
The explicit form of these functions is rather involved and will not be given here.
However, for the construction of the Jost solutions, it suffices to analyze the
asymptotic form of these functions as~$u \rightarrow \pm \infty$.
We recall the basic idea for~$X^{(\omega)}_{+}$ and the asymptotics~$u \rightarrow +\infty$.
We rewrite \eqref{secondOrder} in the form
\[
    \bigg[-\partial^2_u - \Omega_+^2 \bigg]\Jost{+} = - W^+(u)\Jost{+}
\]
with a complex constant~$\Omega_+$ and a potential~$W^+$.
Having the complex coefficients, the differential operator on the left side can be inverted with
the help of an explicitly given Green's kernel~$S(u,v)$, i.e.\
\[
    \big[\partial_v^2 - \Omega_+^2 \big]S(u,v) = \delta(u-v)
\]
This makes it possible to formulate a Lippmann-Schwinger equation of the form
\be
    \label{lippmann}
    \Jost{+} = e^{i \Omega_+ u} + \int_u^{\infty} S(u,v)\:W^+(v)\:\Jost{+}(v)\: \dif{v}
\ee
Using that~$W$ decays as~$u \rightarrow \infty$, one can perform an expansion in powers of~$W$,
\beq \label{Jost+}
    \Jost{+} = \sum_n^{\infty}\Jost{+}^{(n)} \:.
\eeq
Similar to the usual Picard-Lindel\"of iteration (defined on a bounded interval),
one can show that this series converges uniformly,
giving rise to a unique solution with prescribed asymptotics as~$u \rightarrow \infty$.
Substituting the resulting solution~$X_+^{(\omega)}$ into the first-order system~\eqref{radialEq},
one can solve for~$X_-^{(\omega)}$. We thus obtain a solution of the radial equation
with prescribed asymptotics.

This method has been worked out for the Dirac and wave equations in the Kerr geometry
in Boyer-Lindquist coordinates in~\cite{wdecay, schdecay} as well in Eddington-Finkelstein coordinates
in~\cite{hamilton} (for basics and more details see also~\cite{alfaro+regge}).
For simplicity, we state the result only in the region~$(r_+, \infty)$ outside the event horizon
and note that the regions~$(r_0, r_-)$ and~$(r_-, r_+)$ are treated similarly.
\begin{Lemma}
    \label{jost}
For every angular momentum mode $k$ and $l$ and  $\omega_\varepsilon \in
\{ \omega_\varepsilon \neq 0 \, | \, \omega_\varepsilon \in \C\}$,
there are unique radial Jost solution
to the complexified Schrödinger-type equations~\eqref{secondOrder}
in the region~$(r_+, \infty)$
with the asymptotic boundary conditions
\be
\lim_{u\rightarrow +\infty} e^{-i \Omega_+(\omega_\varepsilon)\, u - i c(\omega_\varepsilon) \ln u}\, \mathcal{J}_{+}(u) &= 1 \:, &
\lim_{u\rightarrow -\infty} e^{- i \Omega_-(\omega_\varepsilon)\,  u} \,\mathcal{J}_{-}(u) &= 1 \label{asy1} \\
\lim_{u\rightarrow +\infty} \partial_{u} \Big( e^{-i \Omega_+(\omega_\varepsilon)\,  u - i c(\omega_\varepsilon) \ln u} \mathcal{J}_{+}(u)\Big) &= 0 \;, &
\lim_{u\rightarrow -\infty} \partial_{u} \Big( e^{- i \Omega_-(\omega_\varepsilon)\,  u}\mathcal{J}_{-}(u)\Big) &= 0 \label{asy2}
\ee
for suitable complex numbers~$c$ and~$\Omega_\pm$.
Furthermore, these solutions are analytic and smooth in~$u$ and~$\omega_\varepsilon$.
\end{Lemma}
\begin{proof}
In the case~$\im \omega_\varepsilon > 0$, the Jost solutions have been constructed in~\cite{wdecay, schdecay}.
By complex conjugation of the Schrödinger-type equation, we obtain corresponding solutions
for~$\im \omega_\varepsilon < 0$.
\end{proof} \noindent
We note that the asymptotics~\eqref{asy1} and~\eqref{asy2} for~$\mathcal{J}_{\pm}$ can also be written as
\[ \mathcal{J}_{+}(u) \simeq e^{i \Omega_{+}(\omega_\varepsilon) u + i c(\omega_\varepsilon) \ln u} + E_+(u) \quad \text{and} \quad
\mathcal{J}_{-}(u) \simeq e^{i \Omega_{-}(\omega_\varepsilon) u} +E_-(u) \]
with a decaying error term~$E_{\pm}(u)$. This shows that~$\Omega(\omega_\varepsilon)_{\pm}$ and~$c(\omega_\varepsilon)$ encode the asymptotic phase and
amplitude of the wave.

We now collect the Jost solutions in the respective regions~\eqref{regions}
and introduce a convenient notation. We can use the same notation for the Jost solutions
constructed near the event horizon from inside and outside (and similarly near the Cauchy horizon).
because these solutions have the same asymptotics in the Regge-Wheeler coordinates
of the respective spacetime region. Note that the Regge-Wheeler coordinate~$u$
tends to $-\infty$ at the event horizon and to $+\infty$ at the Cauchy horizon.

In order to find Jost solutions, we must ensure that the integral in~\eqref{lippmann} and the
series~\eqref{Jost+} converge. Depending on the signs of~$\im{(\omega_\varepsilon)}$ and $\re{(\omega_\varepsilon)}$,
we thus obtain different solutions as compiled in the next lemma.
\begin{Lemma}
    \label{defJost}
    We introduce Jost solutions with the following asymptotics:
    \begin{enumerate}[leftmargin=2em]
        \item[{\rm{(i)}}] Near spatial infinity in the case $|\omega| > m$ and sufficiently small~$\varepsilon>0$:
            \be
                &\widehat{\mathcal{J}}_{\infty}(u) = f_{\infty,1} \:U_{\omega_\varepsilon}\:
                \begin{pmatrix}
                    1 \\
                    0
                \end{pmatrix}
                e^{i\phi_+(u)} \,\Bigg[ 1 + \mathcal{O}\Bigg(\dfrac{1}{u}\Bigg) \Bigg] \label{asy414} \\
                &\qquad \text{if }  \im (\omega_\varepsilon) < 0 \, \text{ with } \, \re (\omega_\varepsilon) < 0
                     \, \text{ or }\, \im (\omega_\varepsilon) > 0 \, \text{ with } \, \re (\omega_\varepsilon) > 0 \nonumber \\
                &\widecheck{\mathcal{J}}_{\infty}(u) = f_{\infty,2}\:U_{\omega_\varepsilon}\:
                \begin{pmatrix}
                    0 \\
                    1
                \end{pmatrix}
                e^{-i\phi_-(u)}  \,\Bigg[ 1 + \mathcal{O}\Bigg(\dfrac{1}{u}\Bigg) \Bigg] \label{asy416}  \\
            &\qquad \text{if } \im (\omega_\varepsilon) > 0 \, \text{ with } \, \re (\omega_\varepsilon) < 0
                    \, \text{ or } \, \im (\omega_\varepsilon) < 0 \, \text{ with } \, \re (\omega_\varepsilon) > 0 \nonumber \\
\intertext{with}
            \;\;\;&\,\,\phi_{\pm}(u) = \sqrt{\omega_\varepsilon^2 - m^2}\: u + M \bigg(\pm 2\omega_\varepsilon +
                    \dfrac{m^2}{\sqrt{\omega_\varepsilon^2 - m^2}}\bigg)\ln(u) \notag \\
\intertext{and}
\;\;\;&\,\,U_{\omega_\varepsilon} = \begin{bmatrix}
            \cosh(\Theta) & \sinh(\Theta) \\
            \sinh(\Theta) & \cosh(\Theta) \\
        \end{bmatrix} \quad \text{with} \qquad
\Theta := \frac{1}{4}\, \ln \Big( \frac{\omega_\varepsilon - m}{\omega_\varepsilon + m} \Big) \:. \notag
\ee
\item[{\rm{(ii)}}] Near spatial infinity in the case $|\omega| < m$ and for sufficiently small~$\varepsilon>0$:
        \be
        &\widehat{\mathcal{J}}_{\infty}(u) =  V_{\omega_\varepsilon} \:
        \begin{pmatrix}
            f_{\infty} \\
            0
        \end{pmatrix}
        e^{- \sqrt{m^2 - \omega_\varepsilon^2} \, u} \,\Bigg[ 1 + \mathcal{O}\Bigg(\dfrac{1}{u}\Bigg) \Bigg] \notag \\
        \intertext{with}
        &V_{\omega_\varepsilon} = \begin{bmatrix}
            \dfrac{i m}{2 \sqrt{m^2 - \omega^2}} & \dfrac{1}{2} \big(1 + \dfrac{i \omega}{\sqrt{m^2 - \omega^2}}\big)\\
            - \dfrac{i m}{2 \sqrt{m^2 - \omega^2}} & \dfrac{1}{2} \big(1 - \dfrac{i \omega}{\sqrt{m^2 - \omega^2}}\big)
        \end{bmatrix} . \notag
        \ee
    \item[{\rm{(iii)}}] Near the event horizon $r_+$:
\[            \widecheck{\mathcal{J}}_{+}(u) =
                \begin{pmatrix}
                    0\\
                    h_{+,1}
                \end{pmatrix}
                \Bigg[ 1 + \mathcal{O} \Bigg(e^{bu}\bigg) \Bigg] \quad \text{and} \quad
            \widehat{\mathcal{J}}_{+}(u) =
                \begin{pmatrix}
                    h_{+,2}\\
                    0
                \end{pmatrix} e^{2i\omega_\varepsilon u} \Bigg[ 1 + \mathcal{O} \Bigg(e^{bu}\bigg) \Bigg]
\]
        where, depending on $\im(\omega)$, the above functions are well-defined and in $L^2_{\text{\rm{loc}}}$
        near the event horizon:
        \be
            \widecheck{\mathcal{J}}_{+}(u), \quad \widehat{\mathcal{J}}_{+}(u) \quad &\text{for} \quad \im(\omega_\varepsilon) < 0 \nonumber \\
            \widecheck{\mathcal{J}}_{+}(u) \qquad &\text{for} \quad \im(\omega_\varepsilon) > 0 \nonumber
        \ee

    \item[{\rm{(iv)}}] Near the Cauchy horizon $r_-$:
\[
            \widecheck{\mathcal{J}}_{-}(u) =
                \begin{pmatrix}
                    0\\
                    h_{-,1}
                \end{pmatrix}
            \Bigg[ 1 + \mathcal{O}\Bigg(e^{-bu}\Bigg)\Bigg] \text{ and }
            \widehat{\mathcal{J}}_{-}(u) =
                \begin{pmatrix}
                    h_{-,2}\\
                    0
                \end{pmatrix}
                e^{2i\omega_\varepsilon u}
            \Bigg[1 + \mathcal{O}\Bigg(e^{-bu}\Bigg)\Bigg] \]
        where, depending on $\im(\omega)$, the above functions are well-defined and in~$L^2_{\text{\rm{loc}}}$
        near the Cauchy horizon:
        \be
            \widecheck{\mathcal{J}}_{-}(u), \quad \widehat{\mathcal{J}}_{-}(u) \quad &\text{for} \quad \im(\omega_\varepsilon) > 0 \nonumber \\
            \widecheck{\mathcal{J}}_{-}(u) \qquad &\text{for} \quad \im(\omega_\varepsilon) < 0 \nonumber
        \ee
    \end{enumerate}
Here~$f_{\infty, 1/2} \neq 0$ and $ h_{\pm, 1/2} \neq 0$ are constants and $b \in \R_+$.

All these Jost solutions converge locally uniformly as~$\varepsilon \searrow 0$
respectively~$\varepsilon \nearrow 0$ to smooth solutions of the ODE with~$\varepsilon=0$.
\end{Lemma} \noindent
We remark that, in the limit~$\varepsilon \searrow 0$, the asymptotics of the solutions
near infinity and near the horizons was worked out in~\cite[Theorem~1.1]{krpoun-mueller}.
We also point out that, in the case~$|\omega_\varepsilon | < m$,
there is only one Jost solution, denoted by~$\widehat{\mathcal{J}}_{\infty}$,
which decays exponentially at infinity. In this case, we choose an arbitrary fundamental
solution which is linearly independent of~$\widehat{\mathcal{J}}_{\infty}$ and denote
it by~$\widecheck{\mathcal{J}}_{\infty}$. This fundamental solution increases exponentially near
infinity. It is not canonical. But our main theorem will not depend on this choice.

\begin{proof}[Proof of Lemma~\ref{defJost}]
We begin with the Jost solutions near spatial infinity.
According to Lemma~\ref{krpMul}, the solutions have the plane wave asymptotics
with asymptotic phases given by~\eqref{asyPhase}.
It suffices to consider the linear term in~\eqref{asyPhase}, because for large~$u$ it dominates the logarithm.
In the case~$|\omega_\varepsilon| > |m|$, by a Taylor expansion of the square root we obtain
\[ \sqrt{\omega_\varepsilon^2 - m^2} =
            \sqrt{\omega^2 - m^2}  + i \im{(\omega_\varepsilon)}
            \re{(\omega_\varepsilon)}\dfrac{1}{\sqrt{\omega^2 - m^2}} +
            \mathcal{O}\big(\varepsilon ^2\big) \:. \]
    Since the imaginary term determines the decay behavior, it suffices to consider the second term.
    This shows that the combinations of real and imaginary parts in~\eqref{asy414} and~\eqref{asy416}
    ensure that the exponentials~$e^{\pm i \phi_\pm(u)}$ decay at infinity. This property is precisely what is
    needed in order for the Jost solutions to be well-defined.
In the remaining case~$|\omega_\varepsilon| \leq |m|$, using that~$\sqrt{\omega_\varepsilon^2 - m^2} =  i\sqrt{m^2 - \omega_\varepsilon^2}$ and doing the
same expansion in $\varepsilon$, we see the first term in $u$ dominates the convergent behavior. It is
also independent of $\im (\omega_\varepsilon)$. Therefore, we only have one exponential decaying solution for this case at infinity.\\
    Now, we look at the Jost functions near the event and Cauchy horizons. We only consider the Cauchy horizon, because the event horizon can be treated similarly.\\
    Near the Cauchy horizon, the fundamental solutions have the plane wave
    asymptotics as worked out in Theorem~\ref{krpMul}.
    Using that the potential~$W$ in the Lippmann-Schwinger equation~\eqref{lippmann} decays exponentially, it turns out that there are two convergent Jost
    solutions~$\widehat{\mathcal{J}}_{-}$ and~$\widecheck{\mathcal{J}}_{-}(u)$
    defined for~$\omega_\varepsilon$ close to the real axis, which form a fundamental system
    (for details see~\cite[Section~3]{wdecay}).
\end{proof}
Finally, the boundary condition at $r=r_0$ as initial data for the radial ODE, we obtain
a solution~$\mathcal{J}_{\partial M}$ with the following boundary values.
\begin{enumerate}[leftmargin=2em]
    \item[{\rm{(iv)}}] At the boundary~$\partial M$:
    \bes
        \mathcal{J}_{\partial M}(u) = \mathcal{J}^{(1)}_{\partial M}(u)
        \begin{pmatrix}
            1 \\[0.3em] \displaystyle
            \frac{\sqrt{|\Delta|}}{r_+}
        \end{pmatrix}
    \ees
\end{enumerate}

\subsection{Construction of the Green's Matrix} \label{constructGreen}
Our next goal is to express the Green's matrix~\eqref{eqDistributional} in terms of Jost solutions.
To this end, we make the general ``gluing ansätze'' depending on $r'$ and $\im(\omega_\varepsilon)$
\begin{align*} 
G(r,r') &= \Theta(r-r') \\
&\quad \times \big[ \Phi_1(r) \otimes P_1(r') + \Phi_2(r) \otimes P_2(r') \big] &&\text{ for } r' \in (r_-, r_+)
    \text{ and } \im(\omega_\varepsilon) < 0 \nonumber\\
G(r,r') &= \Theta(r'-r) \\
&\quad \times \big[ \Phi_1(r) \otimes P_1(r') + \Phi_2(r) \otimes P_2(r') \big] &&\text{ for } r' \in (r_-, r_+)
    \text{ and } \im(\omega_\varepsilon) > 0 \nonumber\\
G(r,r') &= \Theta(r-r') \Phi_1(r) \otimes P_1(r') \\
&\qquad + \Theta(r'-r) \Phi_2(r) \otimes P_2(r') &&\text{ for } r' \notin (r_-, r_+)\:,
\end{align*}
where~$\Phi_1$ and~$\Phi_2$ are the Jost solutions defined in the respective regions
\[ \left\{ \begin{array}{cl} \Phi_1(r) & \text{for~$r_0 \leq r  \leq r'$} \\[0.1em]
\Phi_2(r) & \text{for~$r' \leq r < \infty$}\:. \end{array} \right. \]
Here and in what follows, it is more convenient to work again with the radial variable~$r \in [r_0, \infty)$.
The corresponding Regge-Wheeler coordinate is obtained in the respective regions
(i.e., inside the Cauchy horizon, between the horizons, and in the asymptotic end)
by~\eqref{reggeWheeler}. We point out that when extending the solutions~$\Phi_1$ or~$\Phi_2$
across the event or Cauchy horizons, we need to make sure that their $L^2$-norms are finite, i.e.\
\[ \Phi_1 \in L^2((r_0, r'), \C^2) \qquad \text{and} \qquad \Phi_2 \in L^2((r', \infty), \C^2) \:, \]
where we work with the $L^2$-scalar product in~\eqref{skalarprodukt}.

We now explain in detail how the functions~$\Phi_1$ and~$\Phi_2$ can be chosen.
We consider the cases separately when~$\omega_\varepsilon$ is in the upper and lower half plane.
\begin{enumerate}[leftmargin=2em]
    \item[{\rm{(i)}}]
    $\im (\omega_\varepsilon) < 0 \text{ and } \re(\omega_\varepsilon) < 0 \text{ for }
    |\omega_\varepsilon| > m$: \\
    We begin with the case that~$r'$ is outside the event horizon. Due to the choice of of $\omega_\varepsilon$, $\widehat{\Jost{\infty}}(r)$ decays as~$r \rightarrow \infty$, and we can set
    \[
            \Phi^{\infty}_{1}(r) = a \:\widehat{\Jost{\infty}}(r) \:.
    \]
    For $r<r'$ we have two possibilities. One is to take the solution~$\widehat{\Jost{+}}$. This solution
    decays exponentially at the event horizon. Extending it by zero up to~$r=r_0$ gives a solution in~$L^2$.
    This leads us to the ansatz
\beq \label{Phi2}
        \Phi^{\infty}_{2}(r) = b \:\Theta(r'-r) \:\Theta(r - r_+) \:\widehat{\Jost{+}}(r) \:.
\eeq
The other possibility is to take another fundamental solution and to extend it across the
event and Cauchy horizons up to the boundary at~$r=r_0$. Because of the negative imaginary part of $\omega_\varepsilon$, only the
constant function $\widecheck{\mathcal{J}}_{-}(r)$ is square integrable at the Cauchy horizon.
However, for all values of~$\omega$, this solution does not satisfy the
boundary conditions at~$r_0$. For this reason, we are forced to choose the fundamental solution~\eqref{Phi2},
leaving us with two free parameters~$a , b \, \in \C$ (here the superscript~$\infty$ clarifies that~$r'$ lies in the asymptotic end).
These solutions are shown in Figure~\ref{fig1}.
\begin{figure}
    \psscalebox{1.0 1.0} 
    {
    \begin{pspicture}(0,-4.01)(15.42,1.23)


    \definecolor{grey}{rgb}{0.7019608,0.7019608,0.7019608}
    \definecolor{purp}{rgb}{0.3019608,0.2,0.6}

    \psline[linecolor=black, linewidth=0.04](0.00,-3.98)(15.41,-3.98)
    \psline[linecolor=black, linewidth=0.04](1,-3.98)(1,-1.00)
    \psline[linecolor=black, linewidth=0.04, linestyle=dotted, dotsep=0.10583334cm](4,-3.98)(4,-1.00)
    \psline[linecolor=black, linewidth=0.04, linestyle=dotted, dotsep=0.10583334cm](10,-3.98)(10,-1.00)
    \psline[linecolor=grey, linewidth=0.04, linestyle=dashed, dash=0.17638889cm 0.10583334cm](13, -3.98)(13,-1)



    \pszigzag[coilarm=0, coilwidth=0.20, linecolor = blue ](3,-1.90)(5,-1.90)
    \rput[bl](3.25,-1.65){$\widecheck{\Jost{-}}$}

    \pszigzag[coilarm=0, coilwidth=0.20, linecolor = orange ](1,-2.2)(2,-2.2)
    \rput[bl](1.25,-3){$\Jost{\partial M}$}

    \pszigzag[coilarm=0, coilwidth=0.20, linecolor = red ](10,-3.98)(11,-3.5)
    \rput[bl](10.25,-3.25){$\widehat{\Jost{+}}$}

    \pszigzag[coilarm=0, coilwidth=0.20, linecolor = brown ](14.5,-2.5)(15.5,-3)
    \rput[bl](14.5,-3.5){$\widehat{\Jost{\infty}}$}
    \psplot[algebraic=true, linecolor=green]{13}{14}{2 / (x - 12)^2 - 2.98}
    \psplot[algebraic=true, linecolor=green]{11.5}{13}{2 / (x - 14)^2 - 3.5}

    \psline[linecolor=green, linewidth=0.05](1,-3.98)(10,-3.98)

    \psline[linecolor=black, linewidth=0.04, arrowsize=0.05291667cm 2.0,arrowlength=1.4,arrowinset=0.0]{<-}(2.5,-1.75)(3,-0.75)
    \rput[bl](3.2,-0.5){\makebox[0pt]{only non-generic choice}}

    \psline[linecolor=black, linewidth=0.04, arrowsize=0.05291667cm 2.0,arrowlength=1.4,arrowinset=0.0]{<-}(12.8,-1.25)(12.3,-0.75)
    \rput[bl](12.5,-0.5){\makebox[0pt]{discontinuity}}

    \rput[bl](0.9,-4.4){$r_0$}
    \rput[bl](3.9,-4.4){$r_-$}
    \rput[bl](9.9,-4.4){$r_+$}
    \rput[bl](13,-4.4){$r^{\prime}$}
    \rput[bl](15,-4.4){$\longrightarrow \infty$}
    \end{pspicture}
    }
    \vspace{2pt}
\caption{Extending Jost functions across the horizons for $r_+ < r' < \infty$. At $r_-$ no generic linear combinations are possible.
    The green lines indicates the extensions left and right of $r'$. Note that due to the symmetric behavior of the Regge-Wheeler coordinate
    we use the same notation for the constructed Jost solutions from inside and outside the horizons (see Lemma~\ref{jost} and~\ref{defJost}).}
    \label{fig1}
\end{figure}%

We next consider the case~$r_- < r' \leq r_+$. By the same argumentation as above, there are no non-trivial
solutions crossing the Cauchy horizon.
Therefore, the wave needs to vanish at $r'$ and is continuously extended by zero to $r_0$.
This leads us to the above ansatz where both fundamental solutions are considered only in the region~$r>r'$.

For~$r>r_+$, on the other hand, we have two possible functions. Since $\widecheck{\Jost{+}}$ is a constant function it
needs to be the same inside and outside the event horizon restricting $c = c'$, but the decaying solution can have two different
prefactors in the linear combination for waves. This gives us enough freedom to match the solutions.
The different fundamental solutions are illustrated in Figure~\ref{fig2}.
We begin with the most general linear combination and simplify it afterward.
    It has the form
\begin{align}
        \Phi^{r_+}(r) = \Theta(r - r')\bigg[ &c_L \, \Theta(r_+ - r ) \widehat{\Jost{+}}(r) + c \, \Theta(r_+ - r)
         \widecheck{\Jost{+}}(r)  \nonumber \\
        &+ c' \, \Theta(r-r_+) \:\widecheck{\Jost{+}}(r) + c_R \, \Theta(r-r_+)\widehat{\Jost{+}}(r) \nonumber \\
        &+ a \, \Theta(r - r_+) \:\widehat{\Jost{\infty}}(r)\bigg] \:. \label{matching}
\end{align}
    \begin{figure}
        \psscalebox{1.0 1.0} 
        {
        \begin{pspicture}(0,-4.01)(15.42,1.23)

        \definecolor{grey}{rgb}{0.7019608,0.7019608,0.7019608}
        \definecolor{purp}{rgb}{0.3019608,0.2,0.6}

        \psline[linecolor=black, linewidth=0.04](0.00,-3.98)(15.41,-3.98)
        \psline[linecolor=black, linewidth=0.04](1,-3.98)(1,-1.00)
        \psline[linecolor=black, linewidth=0.04, linestyle=dotted, dotsep=0.10583334cm](4,-3.98)(4,-1.00)
        \psline[linecolor=black, linewidth=0.04, linestyle=dotted, dotsep=0.10583334cm](10,-3.98)(10,-1.00)
        \psline[linecolor=grey, linewidth=0.04, linestyle=dashed, dash=0.17638889cm 0.10583334cm](6, -3.98)(6,-1)

        \pszigzag[coilarm=0, coilwidth=0.20, linecolor = red ](10,-3.98)(11,-2.98)
        \pszigzag[coilarm=0, coilwidth=0.20, linecolor = red ](10,-3.98)(9,-2.98)
        \pszigzag[coilarm=0, coilwidth=0.20, linecolor = red ](9,-1.75)(11,-1.75)
        \rput[bl](10.25,-1.5){$c$}
        \rput[bl](9.5,-1.5){$c'$}
        \rput[bl](10.25,-3.25){$c_R$}
        \rput[bl](9.5,-3.25){$c_L$}
        \rput[bl](8.5,-1.9){$\widecheck{\Jost{+}}$}
        \rput[bl](8.6,-3.7){$\widehat{\Jost{+}}$}

        \pszigzag[coilarm=0, coilwidth=0.20, linecolor = blue ](3,-1.90)(5,-1.90)
        \rput[bl](3.25,-1.65){$\widecheck{\Jost{-}}$}

        \pszigzag[coilarm=0, coilwidth=0.20, linecolor = orange ](1,-2.2)(2,-2.2)
        \rput[bl](1.25,-3){$\Jost{\partial M}$}

        \pszigzag[coilarm=0, coilwidth=0.20, linecolor = brown ](13,-2.5)(14.5,-3.5)
        \rput[bl](14.5,-3.0){$\widehat{\Jost{\infty}}$}

        \psline[linecolor=green, linewidth=0.05](1,-3.98)(6,-3.98)
        \psline[linecolor=green, linewidth=0.05](6,-3.2)(8,-2,25)
        \psline[linecolor=green, linewidth=0.05](11.25,-1.75)(12.5,-2.0)

        \psline[linecolor=black, linewidth=0.04, arrowsize=0.05291667cm 2.0,arrowlength=1.4,arrowinset=0.0]{<-}(2.5,-1.75)(3,-0.75)
        \rput[bl](3.2,-0.5){\makebox[0pt]{only non-generic choice}}

        \rput[bl](0.9,-4.4){$r_0$}
        \rput[bl](3.9,-4.4){$r_-$}
        \rput[bl](9.9,-4.4){$r_+$}
        \rput[bl](5.9,-4.4){$r^{\prime}$}
        \rput[bl](15,-4.4){$\longrightarrow \infty$}
        \end{pspicture}
        }
        \vspace{2pt}
        \caption{Extending Jost functions across the horizons for $r_- < r' \leq r_+$. At $r_+$ generic linear combinations are possible
        which is not the case at $r_-$ with only one Jost function. The green lines indicates the extensions left and right of $r'$.
        The zigzag lines imply possible Jost functions. Note that $\widehat{\Jost{+}}$ stands for two possible
        Jost solutions for $\im(\omega_\varepsilon) < 0$\\ (see Lemma~\ref{defJost}).}
        \label{fig2}
    \end{figure}%
    Now we can express $\widehat{\Jost{\infty}}(r)$ as a linear combination of $\widecheck{\Jost{+}}(r)$ and $\widehat{\Jost{+}}(r)$.
    This results in
    \be
        \Phi^{r_+}(r) = \Theta(r-r')\bigg[&c \, \Theta(r_+ - r) \widecheck{\Jost{+}}(r) +
                a \, \Theta(r - r_+) \widehat{\Jost{\infty}}(r)\nonumber \\
                &+ c_L \:\Theta(r_+ - r) \widehat{\Jost{+}}(r) \bigg] \nonumber
    \ee
with two free parameters~$a, c_L \in \C$ (the parameter~$c$, on the other hand, is determined by the
matching conditions on the event horizon). Thus, we have the solutions
\begin{align*}
    \Phi^{r_+}_1(r) &= c_L \:\Theta(r_+ - r) \widehat{\Jost{+}}(r) \\
    \Phi^{r_+}_2(r) &= c \, \Theta(r_+ - r) \widecheck{\Jost{+}}(r) +
        a \, \Theta(r - r_+) \widehat{\Jost{\infty}}(r)
\end{align*}

It remains to consider the case~$r_0 \leq r' \leq r_-$. On the left side of $r'$ the boundary
conditions at~$r_0$ admit, up to a multiple, a unique solution, i.e.
\[ \Phi_1^{r_-}(r) = e\: \mathcal{J}_{\partial M}(r) \:. \]
For~$r>r_-$, only the solution~$\widecheck{\mathcal{J}}_{-}(r)$ can be extended
in $L^2$ across the Cauchy horizon. This solution must be extended
also across the event horizon and must be matched to the decaying solution
(see Figure~\ref{fig3} for a detailed sketch).
Now the matching across the Cauchy horizon determines the constants~$c_L$ and~$c$,
whereas the matching across the event horizon determines~$a$ and~$c_R$. We conclude that
\begin{align*}
        \Phi^{r_-}_{1}(r) &= e \, \Jost{\partial M}(r) \notag \\
        \Phi^{r_-}_{2}(r) &=
        d \, \Theta(r-r_-)\:\widehat{\Jost{-}}(r) + c_L \, \Theta(r-r_-)\: \Theta(r_+ - r ) \widehat{\Jost{+}}(r) \notag \\
        &+ c_R \, \Theta(r - r_+) \widehat{\Jost{+}}(r) + c \, \Theta(r-r_-)\, \Theta(r_+ - r)\:
        \widecheck{\Jost{+}}(r) \notag \\
                &+ a \, \Theta(r - r_+) \widehat{\Jost{\infty}}(r)
\end{align*}
\begin{figure}
    \psscalebox{1.0 1.0} 
    {
    \begin{pspicture}(0,-4.01)(15.42,1.23)

    \definecolor{grey}{rgb}{0.7019608,0.7019608,0.7019608}
    \definecolor{purp}{rgb}{0.3019608,0.2,0.6}

    \psline[linecolor=black, linewidth=0.04](0.00,-3.98)(15.41,-3.98)
    \psline[linecolor=black, linewidth=0.04](1,-3.98)(1,-1.00)
    \psline[linecolor=black, linewidth=0.04, linestyle=dotted, dotsep=0.10583334cm](4,-3.98)(4,-1.00)
    \psline[linecolor=black, linewidth=0.04, linestyle=dotted, dotsep=0.10583334cm](10,-3.98)(10,-1.00)
    \psline[linecolor=grey, linewidth=0.04, linestyle=dashed, dash=0.17638889cm 0.10583334cm](2.5, -3.98)(2.5,-1)

    \pszigzag[coilarm=0, coilwidth=0.20, linecolor = red ](10,-3.98)(11,-2.98)
    \pszigzag[coilarm=0, coilwidth=0.20, linecolor = red ](10,-3.98)(9,-2.98)
    \pszigzag[coilarm=0, coilwidth=0.20, linecolor = red ](9,-1.75)(11,-1.75)

    \rput[bl](10.25,-1.5){$c$}
    \rput[bl](9.5,-1.5){$c'$}
    \rput[bl](10.25,-3.25){$c_R$}
    \rput[bl](9.5,-3.25){$c_L$}
    \rput[bl](8.5,-1.9){$\widecheck{\Jost{+}}$}
    \rput[bl](8.6,-3.7){$\widehat{\Jost{+}}$}

    \pszigzag[coilarm=0, coilwidth=0.20, linecolor = blue ](2.5,-1.90)(5,-1.90)
    \rput[bl](4.25,-1.65){$d$}
    \rput[bl](3.25,-1.65){$\widecheck{\Jost{-}}$}

    \pszigzag[coilarm=0, coilwidth=0.20, linecolor = orange ](1,-2.2)(2.5,-2.2)
    \rput[bl](1.25,-3){$\Jost{\partial M}$}


    \pszigzag[coilarm=0, coilwidth=0.20, linecolor = brown ](13,-2.5)(14.5,-3.5)
    \rput[bl](14.5,-3.0){$\widehat{\Jost{\infty}}$}

    \psline[linecolor=green, linewidth=0.05](5.5,-1.90)(8.5,-1.75)
    \psline[linecolor=green, linewidth=0.05](11.25,-1.75)(12.5,-2.0)

    \psline[linecolor=black, linewidth=0.04, arrowsize=0.05291667cm 2.0,arrowlength=1.4,arrowinset=0.0]{<-}(2.6,-1.75)(3.1,-0.75)
    \rput[bl](3.2,-0.5){\makebox[0pt]{discontinuity}}

    \rput[bl](0.9,-4.4){$r_0$}
    \rput[bl](3.9,-4.4){$r_-$}
    \rput[bl](9.9,-4.4){$r_+$}
    \rput[bl](2.5,-4.4){$r^{\prime}$}
    \rput[bl](15,-4.4){$\longrightarrow \infty$}
    \end{pspicture}
    }
    \vspace{2pt}
    \caption{Extending Jost functions across the horizons for $r_0 \leq r' \leq r_-$. Again, t $r_+$ generic linear combinations are possible.
    This is not the case at $r_-$ with only one Jost function. The green lines indicates the extensions left and right of $r'$.
    This time, due to the discontinuity at $r = r'$ non generic linear combinations are possible.}
    \label{fig3}
\end{figure}%
    with two free parameter $d,e \in \C$ (and~$c_L, c_R, c$ and~$a$ fixed by the matching conditions).\\

    \item[{\rm{(ii)}}]
    $\im (\omega_\varepsilon) > 0 \text{ and } \re (\omega_\varepsilon) < 0 \text{ for }
    |\omega_\varepsilon| >  m $:\\
    The procedure to find the fundamental solutions in this case is similar as in the first case.
    One only needs to keep in mind that the roles of the event and Cauchy horizons are
    interchanged, in the sense that at the Cauchy horizon, both fundamental solutions are in $L^2$,
    whereas at the event horizon, one fundamental solution is singular.
\end{enumerate}
The remaining cases are similar. For $|\omega_\varepsilon| < m$ we need to remember that we only have the decaying solution $\Jost{\infty}$ at infinity.
We summarize the results in following lemma.

\begin{Lemma}
    \label{fundSolSum}
    The extended Jost solutions $\Phi_1$ and $\Phi_2$ can be expressed as
	\begin{enumerate}[leftmargin=1em]
        \item[{\rm{(i)}}] $\im (\omega_\varepsilon) < 0 \text{ and } \re(\omega_\varepsilon) < 0 \text{ and }
            |\omega_\varepsilon| > m$:\\[0.1em]
        \begin{itemize}
            \item $\Phi^{\infty}_1(r) = c_1 \, \Theta(r - r_+)\widehat{\Jost{+}}(r)$\\[0.1em]
            $\Phi^{\infty}_2(r) = c_2 \, \widehat{\Jost{\infty}}(r)$ \\

            \item $\Phi^{r_+}_1(r) = c_1 \, \Theta(r_+ - r) \widehat{\Jost{+}}(r) $\\[0.1em]
            $\Phi^{r_+}_2(r) = c_2\, \Theta(r - r_+) \widehat{\Jost{\infty}}(r) + a_1 \, \Theta(r_+ - r)
            \widecheck{\Jost{+}}(r)$ \\
            \item $\Phi^{r_-}_{1}(r) = c_1 \, \Jost{\partial M}(r)$\\[0.1em]
            $\Phi^{r_-}_{2}(r) =  c_2 \, \Theta(r - r_-) \widecheck{\Jost{-}}(r) + a_1 \,\Theta(r - r_-)\Theta(r_+ - r)\widehat{\Jost{+}}\\
                 \text{}\qquad \quad\;\;\;\; + a_2 \, \Theta(r - r_+) \widehat{\Jost{+}}(r) + a_3 \, \Theta(r-r_-)\Theta(r_+-r)\widecheck{\Jost{+}}(r) + a_4 \, \Theta(r - r_+) \widehat{\Jost{\infty}}(r)$ \\
        \end{itemize}
Here~$c_1$ and~$c_2$ are free parameters, and~$a_1, \ldots, a_4$ are constants which depend
on~$c_1$, $c_2$ and~$\omega_\varepsilon$. \\
        \item[{\rm{(ii)}}] $\quad \im (\omega_\varepsilon) > 0 \text{ and } \re (\omega_\varepsilon) < 0 \text{ and }
            |\omega_\varepsilon| >  m $:\\[0.1em]
        \begin{itemize}
            \item $\Phi^{\infty}_1(r) = c_1 \, \Theta(r - r_+)\widecheck{\Jost{+}}(r) + a_1 \, \Theta(r - r_-)\Theta(r_+ - r)\widecheck{\Jost{-}}(r)\\
                \text{}\qquad \quad + a_2 \, \Theta(r - r_-)\Theta(r_+ - r)\widehat{\Jost{-}}(r) + a_3\, \Theta(r_- - r)\widehat{\Jost{-}}(r) + a_4 \, \Theta(r_- - r) \Jost{\partial M}(r)$\\[0.1em]
                $ \Phi^{\infty}_2(r) = c_2 \, \widecheck{\Jost{\infty}}(r)$\\

            \item $\Phi^{r_+}_1(r) = c_1 \, \Theta(r_- - r) \Jost{\partial M} + a_1 \, \Theta(r - r_-) \widecheck{\Jost{-}}(r)$\\
                  $\Phi^{r_+}_2(r) = c_2 \, \Theta(r -r_-) \widehat{\Jost{-}}(r)$\\[0.1em]

            \item $\Phi^{r_-}_1(r) = c_1 \, \Theta(r_- - r)\Jost{\partial M}(r)$\\[0.1em]
                  $\Phi^{r_-}_2(r) = c_2 \, \Theta(r_- -r) \widehat{\Jost{-}}(r)$\\
        \end{itemize}
Here~$c_1$ and~$c_2$ are again free parameters, and~$a_1, \ldots, a_4$ are constants which depend
on~$c_1$, $c_2$ and~$\omega_\varepsilon$. \\
        \item[{\rm{(iii)}}] $\im (\omega_\varepsilon) < 0 \text{ and } \re (\omega_\varepsilon) > 0 \text{ and }
            |\omega_\varepsilon| >  m $:\\[0.2em]
        \begin{itemize}
            \item Similar to the first case, but with $\widehat{\Jost{\infty}}$ and $\widecheck{\Jost{\infty}}$ interchanged when
                $|\omega_\varepsilon| >  m$.\\[0.1em]
        \end{itemize}

        \item[{\rm{(iv)}}] $ \im (\omega_\varepsilon) > 0 \text{ and } \re (\omega_\varepsilon) > 0 \text{ and }
            |\omega_\varepsilon| >  m $:\\[0.2em]
        \begin{itemize}
            \item Similar to the second case, but with $\widehat{\Jost{\infty}}$ and $\widecheck{\Jost{\infty}}$ interchanged when
            $|\omega_\varepsilon| >  m$.\\[0.1em]
        \end{itemize}

        \item[{\rm{(v)}}] $|\omega_\varepsilon| < m$ with all previous combinations of $\re (\omega_\varepsilon)$ and
        $\im (\omega_\varepsilon)$: \\[0.2em]
            \begin{itemize}
                \item All four cases \rm{(i)}, \rm{(ii)}, \rm{(iii)} and \rm{(iv)} are repeated, but with $\widehat{\Jost{\infty}}$ and $\widecheck{\Jost{\infty}}$ replaced by $\Jost{\infty}$.
            \end{itemize}
    \end{enumerate}
\end{Lemma}


\begin{Remark} (matching conditions and weak solutions) {\em{ We now explain how the
matching conditions derived above can be understood from the perspective of weak solutions
of the Dirac equation. We only consider the event horizon, noting that the Cauchy horizon can be
treated similarly. In~\eqref{matching} we began with a general ansatz for the solution.
Our matching conditions stated that this solution must be continuous across the horizon,
meaning that
\beq \label{matchingsep}
c = c' \:,
\eeq
whereas the prefactors~$c_L$ and~$c_R$ of the solution~$\widehat{\Jost{+}}$ can be chosen arbitrarily and independently
inside and outside the event horizon. This is illustrated in Figure~\ref{fig2}.
An alternative method for deriving these matching conditions is to work out the corresponding
Dirac solution~$\psi$ in~\eqref{direq} by inserting the fundamental solutions into the separation
ansatz~\ref{waveFunc} and~\eqref{dirtrans}. Evaluating the Dirac equation weakly across
the event horizon (similar as is done in Schwarzschild and Boyer-Lindquist coordinates in~\cite{FSYperiodic, kerr}),
one would again get~\eqref{matchingsep}. For brevity we omit the details of this computation.

These matching conditions correspond to the following physical picture. The fundamental
solution~$\widecheck{\Jost{+}}$ describes a wave which crosses the event horizon.
Therefore, current conservation gives rise to a matching condition for this solution.
The fundamental solution~$\widehat{\Jost{+}}$, however, describes a wave which
propagates along the event horizon, but does not cross it. Therefore, we
do not get a matching condition. }} \QEDrem
\end{Remark}

\subsection{Computation of the Green's Matrix}
In this subsection,
we will use the global Jost solutions constructed above to compute the Green's matrix.
For scalar Jost solutions, this construction is well-known and uses the conservation of the Wronskian.
In our setting of two-component solutions, we make use instead of the
conservation of the Dirac current~$\Sl \Phi , G^{\mu} \Phi \Sr$ in radial directions within the given region~$(r_-, r_+)$ and $(r_+, \infty)$.
More precisely, we define a Gram matrix, referred to as the {\em{radial flux matrix}}.
We want to highlight that the corresponding local conservation law holds only
after taking the limit $\varepsilon \searrow 0$ (otherwise the currents on different time slices do not cancel each other).
\begin{Lemma}
    \label{lemmaRadialFlux}
In the limit~$\varepsilon \searrow 0$, the radial flux matrix~$h_{ij}=h_{ij}(c_1, c_2, \omega_\varepsilon)$ defined by
    \be \label{414}
        \bra \Phi_i(r), A \, \Phi_j(r) \ket_{\C^2} =: h_{ij} \quad \text{with} \quad i,j \, \in \{1,2\}
    \ee
and the matrix~$A$ given by
    \[
        A = \begin{bmatrix}
            1 & 0 \\
            0 & -\epsilon(\Delta)\\
        \end{bmatrix}
    \]
    is independent of $r$ in the respective regions~$(r_0, r_-)$, $(r_-, r_+)$ and $(r_+, \infty)$.
\end{Lemma}
\begin{proof}
 We begin with the case~$r \in (r_+, \infty)$. By acting with the partial derivative of $r$ we end up with
    \bes
        \Delta(r)\,\partial_r \bra \Phi(r), A \, \Phi(r) \ket_{\C^2} = \Delta(r) \bra \Phi(r) , \underbrace{(V^\dagger A + A V)}_{= 0} \Phi(r) \ket_{\C^2} = 0 \quad \forall\, r\:,
    \ees
    where $V$ is the matrix on the right side of~\eqref{radialEq}. A straightforward calculation gives the result. The steps in the other regions~$r \in (r_0, r_-)$ and~$r \in (r_-, r_+)$ are identical.
\end{proof}
We remark that the conservation of the radial flux~\eqref{414} was first observed in~\cite{FSYperiodic}
and used in order to rule out non-trivial
time-periodic solutions of the Dirac equation in the exterior Reissner-Nordstr\"om
geometry. Here we can use this conservation law
in order to show that the extended Jost solutions~$\Phi_1$ and~$\Phi_2$
constructed in Section~\ref{prepJost} are linearly independent:
\begin{Lemma} \label{lemmaindepend} In the limit~$\varepsilon \searrow 0$,
the two solutions~$\Phi_1$ and~$\Phi_2$ in Lemma~\ref{fundSolSum}
are linearly independent.
\end{Lemma}
\Proof The case~$|\omega|<m$ is trivial, because~$\Phi_1$ is exponentially decreasing at infinity,
whereas~$\Phi_2$ exponentially increasing.
As the remaining cases can be treated similarly, we only consider case~(i) in Lemma~\ref{fundSolSum}.
Then the radial flux of~$\Phi_2$ can be computed asymptotically as~$r \rightarrow \infty$
using~\eqref{asy414}.
The radial flux of~$\Phi_1$, on the other hand, can be computed at the event horizon
using the asymptotics in Lemma~\ref{defJost}~(iii). If~$\Phi_1$ and~$\Phi_2$ were linearly dependent,
these radial fluxes would have opposite signs, a contradiction.
\QED

Following up, we make the ansatz for $G(r;r')$ when $r' \in (r_+, \, \infty)$.
\be
    \label{ansatz1}
\begin{split}
    G(r,r') &:= \dfrac{\Theta(r-r')}{\Delta(r')} \sum_{j = 1}^2 c_{1j}\:\Phi_1(r) \otimes \big(A(r') \Phi_j(r')\big)^\dagger \\
    &\quad\; +
        \dfrac{\Theta(r'-r)}{\Delta(r')} \sum_{j = 1}^2 c_{2j}\:\Phi_2(r) \otimes \big(A(r') \Phi_j(r')\big)^\dagger
\end{split}
\ee
and for $r' \in (r_-, \, r_+)$
\be
    \label{ansatz2}
    G(r,r') :=  \dfrac{1}{\Delta(r')}
    \begin{cases}
        \Theta(r-r') \sum_{i,j = 1}^2 c_{ij}\Phi_i(r) \otimes \big(A(r') \Phi_j(r')\big)^\dagger \quad \text{for} \quad \im{\omega_\varepsilon} < 0 \\
        \Theta(r'-r) \sum_{i,j = 1}^2 c_{ij}\Phi_i(r) \otimes \big(A(r') \Phi_j(r')\big)^\dagger \quad \text{for} \quad \im{\omega_\varepsilon} > 0 \\
    \end{cases}
\ee
with $c_{ij} \in \, \C$. Note that the matrix~$A(r')$ is constant within the respective regions. Thus, in all follow up computations
we drop the~$r'$ dependency. Additionally, we want to highlight that the above tensor notation corresponds to the bra-ket notation
\[
    X \otimes Y^\dagger = \Ket{X}\Bra{Y} .
\]

\begin{Lemma}
    \label{greenCoeff}
    The Green's matrix is well-defined and bounded in the three regions. Additionally, the coefficients $c_{ij}$ for $r' \notin (r_-, \, r_+) $ are given by
\[ c_{ij} = \begin{bmatrix}
            h^{11} & h^{12} \\
            -h^{21} & -h^{22} \\
        \end{bmatrix} , \]
    where $h^{ij}$ is the inverse matrix of $h_{ij}$ from Lemma~\ref{lemmaRadialFlux}. In the case $r' \in (r_-, \, r_+) $,
    one ends up with
\[ c_{ij} =
        \begin{bmatrix}
            h^{11} & h^{12} \\
            h^{21} & h^{22} \\
        \end{bmatrix} \text{ for } \im{\omega_\varepsilon} < 0 \quad \text{and} \quad
        c_{ij} =
        \begin{bmatrix}
            -h^{11} & -h^{12} \\
            -h^{21} & -h^{22} \\
        \end{bmatrix} \text{ for } \im{\omega_\varepsilon} > 0 \:. \]
\end{Lemma}
\begin{proof} We point out that the Gram matrix~$h_{ij}$ is invertible because
the~$\Phi_1$ and~$\Phi_2$ in~\eqref{414} are two fundamental solutions and~$A$ is a regular matrix.
    The proof is a straightforward computation. We want to solve the distributional equation~\eqref{eqDistributional}
    \[
        \rr(\partial_r ; r) \:G(r ; r')_{\omega_\varepsilon} = \delta(r - r') \Idmat_{\C^2}
    \]
    Inserting the ansatz for $r' \notin (r_-, \, r_+) $ we end up with
\begin{align*}
    \rr(\partial_r ; r) \:G(r ; r')_{\omega_\varepsilon} = &\delta(r-r') \bigg[ c_{11} \Phi(r)_1 \otimes \big(A \Phi(r)_1\big)^\dagger +
        c_{12} \Phi(r)_1 \otimes \big(A \Phi(r)_2\big)^\dagger \notag \\
        &- c_{21} \Phi(r)_2 \otimes \big(A \Phi(r)_1\big)^\dagger - c_{22} \Phi(r)_2 \otimes \big(A \Phi(r)_2\big)^\dagger \bigg]
\end{align*}
    Additionally we have following completeness relation on our Hilbert space
\[ \sum_{i,j} h^{ij} \:\Phi(r)_i \otimes \big( A \Phi(r)_j \big)^\dagger = \Idmat_{\C^2} \]
with
\[ h^{ij} = \frac{1}{\det{h}}
        \begin{bmatrix}
            h_{22} & - h_{12} \\
            -h_{21} & h_{11} \\
        \end{bmatrix} \:. \]
Combining both parts gives the result. In the case $r' \in (r_-, \, r_+) $, the steps are identical.

By Lemma~\ref{fundSolSum} we have globally well-defined Jost solutions which are all bounded. Together with the computed coefficients one sees that all Green's matrices are bounded in the corresponding interval for $r'$.
\end{proof}

\section{Integral Representation of the Dirac Propagator} \label{secintrep}
\subsection{Abstract Representation}
We can now state a first step towards our main result.
\begin{Prp}
    \label{mainProp}
    Let $H$ be the Hamiltonian of the Dirac equation in the Reissner-Nordström geometry in Eddington-Finkelstein coordinates.
    Then the corresponding Dirac propagator has the integral representation
\[ X(\tau, r) = \dfrac{1}{2\pi}\int_{\R} e^{-i \omega \tau} \lim_{\varepsilon \searrow 0} \bigg[
                \big(H_\xi - \Idmat_{\C^2}(\omega - i\varepsilon) \big)^{-1} - \big(H_\xi + \Idmat_{\C^2}(\omega - i\varepsilon)\big)^{-1}\bigg] \, X_0(r) \, \dif{\omega}\:. \]
Here the resolvent can be expressed in forms of the Green's matrices $G(r;r')_{\omega_\varepsilon}$ from the ansatz~\eqref{ansatz1} and~\eqref{ansatz2} plus
the coefficients from Lemma~\ref{greenCoeff}. It has the expression
\[ (H_k - \omega \pm i\varepsilon )^{-1} \, X_0(r) =
        \int_{r_0}^{\infty} G(r;r')_{\omega_\varepsilon}\:C(r') \,
        X_0(r') \, \dif{r'} \]
    with
\[ C(r') =
        \begin{bmatrix}
            2r'^2 - \Delta(r')& 0 \\
            0 & -  \Delta(r') \\
        \end{bmatrix} \:, \]
    where $X_0 \in C^\infty_\text{\rm{init}}(N)$ is the initial data for a fixed angular momentum mode $k,l$.
\end{Prp}
\begin{proof}
Combining the results from Propositions~\ref{lemmaResolvente}, \ref{propResolvent} and \ref{greenCoeff},
the remaining task is to interchange the limit~$\varepsilon \searrow 0$ with the integral
and then to take the limit~$a \rightarrow \infty$
in~\eqref{Xresvolent}. Since all extended Jost functions are bounded (for details see Subsection~\ref{constructGreen}), we can apply Lebesgue's dominated convergence theorem
to take the limit~$\varepsilon \searrow 0$ inside the integral.
The limit~$a \rightarrow \infty$ exists by Stone's theorem.
\end{proof}
\subsection{Main Theorem}
By further calculations it is possible to bring the result from Proposition~\ref{mainProp} in a much more handy form. We begin by computing
the Green's matrices in the upper and lower complex plane separately and taking the limit $\varepsilon \searrow 0$ for $r' \in (r_+, \infty)$.
Afterward, we will take the difference and find a more compact expression for the integral representation.
In the end we can extends this to $r' \in (r_-, \infty)$.

\begin{Lemma}
    \label{differenceGreen}
    For $r' \in (r_+ , \infty)$ we can express the differences of the Green's function as
\[ \lim_{\varepsilon \searrow 0 }\,G(r; r')_{>0} - \lim_{\varepsilon \nearrow 0 }\,G(r; r')_{<0}
= \sum_{i,j = 1}^2 \dfrac{g_{ij}}{\Delta(r')} \:\chi_{i}(r) \otimes \chi_{j}(r')^\dagger A \]
with coefficients~$g_{ij}$ of the form
\begin{align}
&g_{11} = g_{22} = 1,\; g_{12} = \frac{a}{b}, \; g_{21} = \frac{d}{c} &&
\text{if~$|\omega|>m$} \label{case1} \\
&g_{ij} = f\, \delta_{i,1} \delta_{j,1} && \text{if~$|\omega|<m$} \:, \label{case2}
\end{align}
where~$a,b,c,d,f \in \C$ and~$b,c \neq 0$.
Moreover, $A$ is again the matrix from Lemma~\ref{lemmaRadialFlux}.
    Additionally, $\chi(r) = (\chi_{1}(r), \chi_2(r))^T$ are the limits of the Jost solutions from Lemma~\ref{defJost}
    defined on the real axis and
    $G(r; r')_{>0}$ describes the Green's matrix on the upper, as well as $G(r; r')_{<0}$ on the lower complex half plane.
\end{Lemma}
\begin{proof}
We begin with the case~$|\omega_\varepsilon|>m$. By looking at Lemma~\ref{fundSolSum} with $r' \in (r_+, \infty)$ for $\im(\omega_\varepsilon) > 0$ we get two extended Jost solutions $\Phi_1(r)$ and $\Phi_2(r)$. Since one can express
any Jost solutions as a linear combination of two others, we choose for $\Phi_2(r)$ a different ansatz to simplify the calculations
\beq \label{Phi12up}
    \Phi_1(r) = \widehat{\Jost{\infty}}(r) \quad \text{and} \quad \Phi_2(r) =  a \: \widehat{\Jost{\infty}}(r) + b \: \widecheck{\Jost{\infty}}(r) \:,
\eeq
where the coefficients of the linear combinations are denoted by~$a, \, b \in \C$.
It follows from Lemma~\ref{lemmaindepend} that~$b$ is non-zero.
First we will calculate the coefficients $c_{ij}$ with those functions. After
that we substitute everything into the ansatz~\eqref{ansatz1} and compute the
result.
We want to highlight that the matrix from lemma~\ref{defJost} is
pseudo-orthonormal to our Wronskian product. Because we are in the complex
planes, we need to treat the product $\la U_{\omega_\varepsilon} | A \, U_{\omega_\varepsilon} \ra$ more carefully. We can re-write the hyperbolic functions in exponential functions and expand them in a power series
\be
    U_{\omega_{\epsilon}} \approx U_{|\omega_\varepsilon|} + i\varphi
    (\varepsilon)
    \begin{pmatrix}
        \sinh(|\omega_\varepsilon|) & \cosh(|\omega_\varepsilon|)\\
        \cosh(|\omega_\varepsilon|) & \sinh(|\omega_\varepsilon|)
    \end{pmatrix}. \notag
\ee
Thus, we end up with
\be
    (U_{\omega_\varepsilon})^\dagger \, A \, U_{\omega_\varepsilon}  \approx
    \begin{bmatrix}
        1 & 0 \\
        0 & -1
    \end{bmatrix} +2\, i \, \varphi(\varepsilon)
    \begin{bmatrix}
        0 & 1 \\
        -1 & 0
    \end{bmatrix} + \varphi(\varepsilon)^2
    \begin{bmatrix}
        -1 & 0 \\
        0 & 1
    \end{bmatrix} \notag .
\ee
We are only interested in the term of zeroth-order in $\varepsilon$ because we
will perform the limit $\varepsilon \rightarrow 0$ in the end. Having that in
mind, we can compute the coefficients for the Greens function with our
ansatz~\ref{Phi12up} (note that in the considered region~$r' \in (r_+, \infty)
$, the function~$\Delta
(r')$ is positive).
\be c_{ij} = -\dfrac{1}{|b|^2}
    \begin{bmatrix}
        |a|^2 - |b|^2 & -a \\
        a^* & -1 \\
    \end{bmatrix} .
\ee
Now we begin by computing the first term from \eqref{ansatz1} (denoted by the superscript~$(1)$),
\bes
    G(r,r')_{> 0}^{(1)} &= \dfrac{\Theta(r-r')}{\Delta(r')} \, \sum_{j = 1}^2 c_{1j}\Phi_1(r) \otimes \big(A \Phi_j(r')\big)^\dagger \notag \\
    & = \dfrac{\Theta(r-r')}{\Delta(r')} \:\bigg[\dfrac{a}{b} \: \widehat{\Jost{\infty}}(r) \otimes \big(A \:\widecheck{\Jost{\infty}}(r')\big)^\dagger
        + \widehat{\Jost{\infty}}(r) \otimes \big(A \:\widehat{\Jost{\infty}}(r')\big)^\dagger \bigg] .
\ees
Repeating the steps for the second term gives
\bes
    G(r,r')_{> 0}^{(2)} &= \dfrac{\Theta(r'-r)}{\Delta(r')} \, \sum_{j = 1}^2 c_{2j}\Phi_2(r) \otimes \big(A \Phi_j(r')\big)^\dagger \notag \\
        & = \dfrac{\Theta(r'-r)}{\Delta(r')} \, \bigg[\widecheck{\Jost{\infty}}(r) \otimes \big(A \:\widecheck{\Jost{\infty}}(r')\big)^\dagger
            +  \dfrac{a}{b} \: \widehat{\Jost{\infty}}(r) \otimes \big(A \:\widecheck{\Jost{\infty}}(r')\big)^\dagger \bigg] .
\ees
Taking the sum of both terms leads to
\bes
    G(r,r')_{> 0} & = \dfrac{1}{\Delta(r')}\, \dfrac{a}{b} \: \widehat{\Jost{\infty}}(r) \otimes \big(A \:\widecheck{\Jost{\infty}}(r')\big)^\dagger \notag \\
    &\quad\: + \dfrac{\Theta(r-r')}{\Delta(r')} \: \widehat{\Jost{\infty}}(r) \otimes \big(A \:\widehat{\Jost{\infty}}(r')\big)^\dagger \notag \\
    &\quad\: + \dfrac{\Theta(r'-r)}{\Delta(r')} \: \widecheck{\Jost{\infty}}(r) \otimes \big(A \:\widecheck{\Jost{\infty}}(r')\big)^\dagger \:.
\ees
Since the resolvent is bounded we can take the $\epsilon$-limit inside the integral from Proposition~\ref{mainProp} and apply the limit on the Jost solutions from the upper and
lower complex plane. More importantly, the Jost solutions from the lower and upper plane coincide on the real axis. We define following limits
\bes
    \lim_{\varepsilon \searrow 0} \widehat{\Jost{\infty}}(r) = \lim_{\varepsilon \nearrow 0} \widehat{\Jost{\infty}}(r) =: \chi_1(r) \quad \text{and} \quad
    \lim_{\varepsilon \searrow 0} \widecheck{\Jost{\infty}}(r) = \lim_{\varepsilon \nearrow 0} \widecheck{\Jost{\infty}}(r) =: \chi_2(r) ,
\ees
Therefore, we end up with the result for the upper half plane
\be \label{sol1}
    \lim_{\varepsilon \searrow 0} \, G(r,r')_{> 0} & = \dfrac{1}{\Delta(r')}\, \dfrac{a}{b} \: \chi_1(r) \otimes \chi_2(r')^\dagger A \notag \\
        &\quad\: + \dfrac{\Theta(r-r')}{\Delta(r')} \: \chi_1(r) \otimes \chi_1(r')^\dagger A \notag \\
        &\quad\: + \dfrac{\Theta(r'-r)}{\Delta(r')} \: \chi_2(r) \otimes \chi_2(r')^\dagger A \, .
\ee

In the lower half plane, we take the Jost solutions
\beq \label{Phi12down}
    \Phi_1(r) = \widecheck{\Jost{\infty}}(r) \quad \text{and} \quad \Phi_2(r) =  c \: \widehat{\Jost{\infty}}(r) + d \: \widecheck{\Jost{\infty}}(r) \qquad \text{if~$|\omega| > m$}\:.
\eeq
with $c, \, d\, \in \C$.
It follows from Lemma~\ref{lemmaindepend} that~$c$ is non-zero.
Repeating similar steps for the lower half plane, we end up with
\bes
    \lim_{\varepsilon \nearrow 0} \, G(r,r')_{< 0} & = - \dfrac{1}{\Delta(r')}\, \dfrac{d}{c} \: \chi_2(r) \otimes \chi_1(r')^\dagger A \notag \\
    &\quad\: - \dfrac{\Theta(r'-r)}{\Delta(r')} \: \chi_1(r) \otimes \chi_1(r')^\dagger A \notag \\
    &\quad\: - \dfrac{\Theta(r-r')}{\Delta(r')} \: \chi_2(r) \otimes \chi_2(r')^\dagger A \, .
\ees
Taking the difference and setting the coefficients as the matrix entries gives the result~\eqref{case1}.

In the case~$|\omega_\varepsilon|<m$, however, we choose the same ansatz~\ref{Phi12up} for the positive, but take a different one for the negative complex plane, i.e.
\be \label{Phi12down2}
    \Phi_1(r) = \widehat{\Jost{\infty}}(r) \quad \text{and} \quad \Phi_2(r) =  a \: \widehat{\Jost{\infty}}(r) + b \: \widecheck{\Jost{\infty}}(r) \qquad \text{if~$|\omega| < m$}\:. \notag
\ee
Furthermore, we have the matrix $V_{\omega_\varepsilon}$ in front of our fundamental solution $\widehat{\Jost{\infty}}(r)$. This matrix behaves differently in the Wronskian and results in a mixing of the components of our fundamental solutions
\be
    (V_{\omega_\varepsilon})^\dagger A \, V_{\omega_\varepsilon} =
    \underbrace{\dfrac{i m}{2 \sqrt{m^2 - \omega^2}}}_{=:g(\omega,\, m)}
    \begin{bmatrix}
        0 & -1 \\
        1 & 0
    \end{bmatrix} + \mathcal{O}(\varepsilon). \notag
\ee
Again, we are only interested in the zero order term in~$\varepsilon$. This time, we end up with a coefficients matrix of the form
\[  c_{ij} = \dfrac{g}{|b|^2}
    \begin{bmatrix}
        a^*b - ab^* & -b \\
        - b* & 0
    \end{bmatrix}, \]
for the upper complex plane. Continuing with similar computations, the Greens
function has the form
\begin{align*}
    G(r,r')_{> 0} &= - \dfrac{g}{\Delta(r')} \dfrac{a}{b}\: \widehat{\Jost{\infty}}(r) \otimes \widehat{\Jost{\infty}}(r')^\dagger A  \\
    &-\dfrac{g \, \Theta(r-r')}{\Delta(r')}\: \widehat{\Jost{\infty}}(r) \otimes \widecheck{\Jost{\infty}}(r')^\dagger A \\
    &- \dfrac{g \, \Theta(r'-r)}{\Delta(r')}\: \widecheck{\Jost{\infty}}(r) \otimes \widehat{\Jost{\infty}}(r')^\dagger A \:.
\end{align*}
A similar computation gives the result for the lower complex plane. Taking the difference and the limit $\varepsilon \rightarrow 0$ gives the second result~\eqref{case2},
\[
    \lim_{\varepsilon \searrow 0} G(r,r')_{> 0} - \lim_{\varepsilon \nearrow 0} G(r,r')_{< 0} = \dfrac{1}{\Delta(r')} \, \underbrace{g \bigg(\dfrac{a}{b} - \dfrac{c}{d}\bigg)}_{=: f} \widehat{\Jost{\infty}}(r) \otimes \widehat{\Jost{\infty}}(r')^\dagger A \:. \]
This concludes the proof.
\end{proof}
In a next step we want to compute the spectral measure from the difference of the Green's matrices only outside the black hole. This is possible because the
Green's matrix is evaluated in the resolvent only pointwise, making it possible to split the integral due to linearity.
\begin{Lemma} \label{lemmaSpec1}
The spectral measure of the Dirac Hamiltonian in Eddington-Finkelstein coordinates for $r' \in (r_+,\infty)$ on initial data~$X_0 \in C^\infty_\text{\rm{init}}(N)$ has the form
\[ \dif{E}_{\omega}(X_0)(r) = \sum_{i,j} t_{ij} \, \chi_i(r)  \otimes \int_{r_+}^{\infty} \chi_j(r')^\dagger \, \Gamma(r') \, X_0(r') \, \dif{r'} \, \dif{\omega} \:, \]
where~$t_{ij}$ are the components of the matrix~$T$ given by
\beq \label{Tmatrix0}
T = \left\{ \begin{array}{cl} \begin{pmatrix}  \displaystyle f & 0 \\
0 & 0 \end{pmatrix} & \text{if~$|\omega|<m$} \\[1em]
\begin{pmatrix} \displaystyle 1 & \dfrac{a}{b} \\
\dfrac{d}{c} & 1 \end{pmatrix} & \text{if~$|\omega|>m$}\:. \end{array} \right.
\eeq
Moreover, $\Gamma := \Gamma \big|_{2 \times 2}$ is the upper left $2\times2$-block of the matrix in~\eqref{scalMat}.
\end{Lemma}
\begin{proof}
    By linearity we can pull the difference of the resolvents and by dominated convergence the $\varepsilon$-limit from Proposition~\ref{mainProp}
    into the integral over $r'$. Then we use the result from Lemma~\ref{differenceGreen} and obtain the result by direct computation.
\end{proof}

We next extend the spectral projection to the region~$r' \in (r_-, \infty)$.
\begin{Lemma}
    \label{lemmaSpec2}
    The spectral measure of the Dirac Hamiltonian in Eddington-Finkelstein coordinates on initial data
    $X_0 \in C^\infty_\text{\rm{init}}(N)$ from Lemma~\ref{lemmaSpec1} extends to $r' \in (r_-, \infty)$, i.e.
\beq \label{Edef}
    \dif{E}_{\omega}(X_0)(r) = \sum_{i,j} t_{ij} \, \chi_i(r)  \otimes  \int_{r_-}^{\infty}
    \chi_j(r')^\dagger \, \Gamma(r') \, X_0(r') \, \dif{r'} \, \dif{\omega} \:.
\eeq
Moreover, the matrix~$T$ in~\eqref{Tmatrix0} can be simplified to
\beq \label{Tmatrix}
T = \left\{ \begin{array}{cl} \begin{pmatrix} \displaystyle f & 0 \\
    0 & 0 \end{pmatrix} \text{ with } f \in \R & \text{if~$|\omega|<m$} \\[0.8em]
\begin{pmatrix} \displaystyle 1 & \dfrac{a}{b} \\
\dfrac{a^*}{b^*} & 1 \end{pmatrix} & \text{if~$|\omega|>m$}\:. \end{array} \right.
\eeq
\end{Lemma}
\begin{proof}
Since the spectral projector is a symmetric operator we can interchange the variables $r \leftrightarrow r'$. Therefore, the relation in
Lemma~\ref{lemmaSpec1} also holds if~$r \in (r_+, \infty)$ and~$r' \in (r_- , \infty)$.
With this in mind, it remains to consider the case~$r,r' \in (r_-, r_+)$.
To this end, one repeats the computational steps in Lemma~\ref{differenceGreen} for
the ansatz~\eqref{ansatz2}, again with the same fundamental solutions~$\Phi_1$ and~$\Phi_2$
in~\eqref{Phi12up} and~\eqref{Phi12down}. A straightforward computation shows that
the spectral projector is again of the form as in~\eqref{Edef}.

It remains to show that~$d/c=a^*/b^*$. To this end, we use that the Hamiltonian is symmetric with respect to the conserved scalar product from Lemma~\ref{lemmasprod}.
    By direct computation, one sees that the resolvent needs to be symmetric with respect to the adjoint operator of the underlying Hilbert space. Thus, the equation
    \[
        \big(\Gamma R_{\omega_\varepsilon} \big)^\dagger = \Gamma R_{\overline{\omega_\varepsilon}}
    \]
    needs to be satisfied, which determines the quotient~$d/c = a^*/b^*$ (here the dagger denotes the adjoint with respect to the scalar product $L^2(\dif{r})$). In the case $|\omega_\varepsilon| < m$ we can use the same argumentation from above. This time, we only have an entry in the diagonal of the matrix $T$. Therefore, $f$ needs to be a real number. This gives the result.
\end{proof}

We are now in the position to state the main theorem of this paper.

\begin{Thm}
    \label{mainTheo}
    The Dirac propagator in Proposition~\ref{mainProp} can be expressed in terms of globally defined fundamental solutions $\chi_i(r,\omega)$ for~$i \in \{1,2\}$ and
    $r\, \in (r_-, \infty)$ as
\[ X(\tau, r) = \dfrac{1}{2} \int_{\R \setminus \{\pm m\} } e^{-i\omega \tau} \, \sum_{i=1}^2 \,\widehat{X}_i(\omega)
                \,\chi_i(r,\omega)\:\dif{\omega} \]
        with $\widehat{X}_i(\omega) : \C \rightarrow \C$ smooth functions defined by
\[ \widehat{X}_i(\omega) = \dfrac{1}{(2\pi)^2}\sum_{j=1}^2 t_{ij}\, (\chi_j(\omega)\,|\, X_0)\:. \]
        Here~$(\cdot| \cdot )$ denotes the conserved scalar product on the hypersurfaces defined in~\eqref{scalPro}
        and given more explicitly in~\eqref{skalarprodukt}, restricted
        to the upper left $2 \times 2$ block. Moreover, $t_{ij}$ are again the entries of the matrix~\eqref{Tmatrix}
        and~$X_0(r) \, \in C^\infty_\text{\rm{init}}(N)$. Here~$a,\,b\, \in \C$
are the transmission coefficients of the radial ODE defined in~\eqref{Phi12up}.
\end{Thm}
\begin{proof}
    We combine the results from Proposition~\ref{mainProp}, Lemma~\ref{lemmaSpec1} and Lemma~\ref{lemmaSpec2} and evaluate the integral over $r'$
    only in the range $(r_-, \infty)$. This is possible because of the boundary conditions the wave is reflected on $\partial M$ and never comes back through the
    Cauchy horizon. Thus all interactions behind the Cauchy horizon do not contribute. Using the expression of the scalar product on the Cauchy hypersurfaces from
    Lemma~\ref{lemmasprod} gives the final expression.
\end{proof}

We finally remark that this integral representation is independent of the choice of the prefactors
of the fundamental solutions. Indeed, as demonstrated in the proof of Lemma~\ref{differenceGreen},
the difference of the Green's matrices depends solely on the quotient of the coefficients of the fundamental solutions.
In this way, the prefactors are drop out.

\Thanks{{{\em{Acknowledgments:}}}
We would like to thank Olaf M\"uller for helpful discussions.
We are grateteful to the referees for the careful reading and valuable feedback.
C.K.\ gratefully acknowledges support by the Heinrich-B\"oll-Stiftung.}

\bibliographystyle{amsplain}

\begin{thebibliography}{10}

\bibitem{bernal+sanchez}
A.N. Bernal and M.~S{\'a}nchez, \emph{On smooth {C}auchy hypersurfaces and
  {G}eroch's splitting theorem},
  \href{https://arxiv.org/abs/gr-qc/0306108}{arXiv:gr-qc/0306108}, Commun.
  Math. Phys. \textbf{243} (2003), no.~3, 461--470.

\bibitem{chandra}
S.~Chandrasekhar, \emph{The {M}athematical {T}heory of {B}lack {H}oles}, Oxford
  Classic Texts in the Physical Sciences, The Clarendon Press Oxford University
  Press, New York, 1998.

\bibitem{chernoff73}
P.R. Chernoff, \emph{Essential self-adjointness of powers of generators of
  hyperbolic equations}, J. Funct. Anal. \textbf{12} (1973), 401--414.

\bibitem{alfaro+regge}
V.~de~Alfaro and T.~Regge, \emph{Potential scattering}, North-Holland
  Publishing Company, Amsterdam, 1965.

\bibitem{kerr}
F.~Finster, N.~Kamran, J.~Smoller, and S.-T. Yau, \emph{Nonexistence of
  time-periodic solutions of the {D}irac equation in an axisymmetric black hole
  geometry}, \href{https://arxiv.org/abs/gr-qc/9905047}{gr-qc/9905047}, Comm.
  Pure Appl. Math. \textbf{53} (2000), no.~7, 902--929.

\bibitem{decay}
\bysame, \emph{Decay rates and probability estimates for massive {D}irac
  particles in the {K}err-{N}ewman black hole geometry},
  \href{https://arxiv.org/abs/gr-qc/0107094}{arXiv:gr-qc/0107094}, Commun.
  Math. Phys. \textbf{230} (2002), no.~2, 201--244.

\bibitem{tkerr}
\bysame, \emph{The long-time dynamics of {D}irac particles in the
  {K}err-{N}ewman black hole geometry},
  \href{https://arxiv.org/abs/gr-qc/000508}{arXiv:gr-qc/0005088}, Adv. Theor.
  Math. Phys. \textbf{7} (2003), no.~1, 25--52.

\bibitem{wdecay}
\bysame, \emph{Decay of solutions of the wave equation in the {K}err geometry},
  \href{https://arxiv.org/abs/gr-qc/0504047}{gr-qc/0504047}, Commun. Math.
  Phys. \textbf{264} (2006), no.~2, 465--503.

\bibitem{sigrn}
F.~Finster and C.~Krpoun, \emph{The fermionic signature operator in the
  {R}eissner-{N}ordstr\"om geometry in {E}ddington-{F}inkelstein coordinates},
  in preparation (2023).

\bibitem{finite}
F.~Finster and M.~Reintjes, \emph{A non-perturbative construction of the
  fermionic projector on globally hyperbolic manifolds {I} -- {S}pace-times of
  finite lifetime}, \href{https://arxiv.org/abs/1301.5420}{arXiv:1301.5420
  [math-ph]}, Adv. Theor. Math. Phys. \textbf{19} (2015), no.~4, 761--803.

\bibitem{chernoff}
F.~Finster and C.~R\"oken, \emph{Self-adjointness of the {D}irac {H}amiltonian
  for a class of non-uniformly elliptic boundary value problems},
  \href{https://arxiv.org/abs/1512.00761}{arXiv:1512.00761 [math-ph]}, Annals
  of Mathematical Sciences and Applications \textbf{1} (2016), no.~2,
  301–--320.

\bibitem{hamilton}
\bysame, \emph{An integral spectral representation of the massive {D}irac
  propagator in the {K}err geometry in {E}ddington-{F}inkelstein-type
  coordinates}, \href{https://arxiv.org/abs/1606.01509}{arXiv:1606.01509
  [gr-qc]}, Adv. Theor. Math. Phys. \textbf{22} (2018), no.~1, 47--92.

\bibitem{schdecay}
F.~Finster and J.~Smoller, \emph{Decay of solutions of the {T}eukolsky equation
  for higher spin in the {S}chwarzschild geometry},
  \href{https://arxiv.org/abs/gr-qc/0607046}{arXiv:gr-qc/0607046}, Adv. Theor.
  Math. Phys. \textbf{13} (2009), no.~1, 71--110, erratum in Adv. Theor. Math.
  Phys. {\bf{20}} (2016), no. 6, 1485--1486.

\bibitem{FSYperiodic}
F.~Finster, J.~Smoller, and S.-T. Yau, \emph{Non-existence of time-periodic
  solutions of the {D}irac equation in a {R}eissner-{N}ordstr\"om black hole
  background}, \href{https://arxiv.org/abs/gr-qc/9805050}{arXiv:gr-qc/9805050},
  J. Math. Phys. \textbf{41} (2000), no.~4, 2173--2194.

\bibitem{goldberg}
J.N. Goldberg, A.J. Macfarlane, E.T. Newman, F.~Rohrlich, and E.C.G. Sudarshan,
  \emph{Spin-{$s$} spherical harmonics and {$\mathchar'26\mkern-12mu
  \partial$}}, J. Math. Phys. \textbf{8} (1967), 2155--2161.

\bibitem{krpoun-mueller}
C.~Krpoun and O.~M\"{u}ller, \emph{Solutions to the {D}irac equation in
  {K}err-{N}ewman geometries including the black-hole region},
  \href{https://arxiv.org/abs/2204.05741}{arXiv:2204.05741 [math.AP]}, J. Geom.
  Phys. \textbf{183} (2023), Paper No. 104689, 19.

\bibitem{lawson+michelsohn}
H.B. Lawson, Jr. and M.-L. Michelsohn, \emph{Spin {G}eometry}, Princeton
  Mathematical Series, vol.~38, Princeton University Press, Princeton, NJ,
  1989.

\bibitem{oneill}
B.~O'Neill, \emph{The {G}eometry of {K}err {B}lack {H}oles}, A K Peters Ltd.,
  Wellesley, MA, 1995.

\bibitem{reed+simon}
M.~Reed and B.~Simon, \emph{Methods of {M}odern {M}athematical {P}hysics. {I},
  {F}unctional analysis}, second ed., Academic Press Inc., New York, 1980.

\end{thebibliography}
\providecommand{\bysame}{\leavevmode\hbox to3em{\hrulefill}\thinspace}
\providecommand{\MR}{\relax\ifhmode\unskip\space\fi MR }
\providecommand{\MRhref}[2]{%
  \href{http://www.ams.org/mathscinet-getitem?mr=#1}{#2}
}
\providecommand{\href}[2]{#2}

\end{document}